\documentclass[reqno]{amsart}

\usepackage{amssymb}
\usepackage{enumerate}
\usepackage{graphicx}

\usepackage{hyperref}

\numberwithin{equation}{section}

\newcommand{\qtx}[1]{\quad\text{#1}\quad}

\newcommand{\beq}{\begin{equation}}
\newcommand{\eeq}{\end{equation}}

\newcommand{\bl}{\begin{lemma}}
\newcommand{\el}{\end{lemma}}
\newcommand{\bc}{\begin{coro}}
\newcommand{\ec}{\end{coro}}
\newcommand{\bp}{\begin{prop}}
\newcommand{\ep}{\end{prop}}
\newcommand{\bd}{\begin{defini}}
\newcommand{\ed}{\end{defini}}

\DeclareMathOperator{\im}{\mathrm{Im}}  
\DeclareMathOperator{\re}{\mathrm{Re}}

\newtheorem{theorem}{Theorem}[section]
\newtheorem{lemma}[theorem]{Lemma}
\newtheorem{coro}[theorem]{Corollary}
\newtheorem{prop}[theorem]{Proposition}
\newtheorem{remark}[theorem]{Remark} 
\newtheorem{defini}[theorem]{Definition}

\newtheorem*{assumptions}{Assumptions}
\newcommand{\de}{\delta}

\newcommand{\lb}{\lambda}
\newcommand{\ze}{\zeta}

\newcommand{\bvp}{{\boldsymbol{\varphi}}}

\newcommand{\BPhi}{\mathbf{\Phi}}
\newcommand{\BPsi}{\boldsymbol{\Psi}}

\newcommand{\ve}{\varepsilon }

\newcommand{\e}{{\rm e}}

\newcommand{\I}{\mathcal{I}}
\newcommand{\Hh}{\mathcal{H}}
\newcommand{\Ff}{\mathcal{F}}
\newcommand{\Kk}{\mathcal{K}}
\newcommand{\Tt}{\mathcal{T}}
\newcommand{\cB}{\mathcal{B}}
\newcommand{\Bb}{\mathcal{B}}
\newcommand{\Aa}{\mathcal{A}}
\newcommand{\Mm}{\mathcal{M}}

\newcommand{\Ll}{\mathcal{L}}
\newcommand{\Tr}{{\rm Tr}}
\newcommand{\Cc}{\mathcal{C}}

\newcommand{\Pp}{\mathcal{P}}
\newcommand{\Ee}{\mathcal{E}}
\newcommand{\Ss}{\mathcal{S}}
\newcommand{\Oo}{\mathcal{O}}

\newcommand{\RR}{\mathbb{R}}
\newcommand{\NN}{\mathbb{N}}
\newcommand{\ZZ}{\mathbb{Z}}

\newcommand{\CC}{\mathbb{C}}
\newcommand{\VV}{\mathbb{V}}
\newcommand{\TT}{\mathbb{T}}
\newcommand{\WW}{\mathbb{W}}
\newcommand{\HH}{\mathbb{H}}

\newcommand{\IS}{\mathbb{T}}

\newcommand{\E}{\mathbb E}
\newcommand{\G}{\mathbb G}

\newcommand{\aaa}{\mathbf{a}}
\newcommand{\bbb}{\mathbf{b}}
\newcommand{\ccc}{\mathbf{c}}

\newcommand{\Sym}{{\rm Sym}}
\newcommand{\diag}{{\rm diag}}
\newcommand{\sgn}{{\rm sgn}}

\newcommand{\hn}{|\!|\!|}

\newcommand{\hnn}{|\!|\!|\!|}

\newcommand{\PE}{\mbox{$\Pp\Ee$}}

\newcommand{\vze}{\vec{\zeta}}
\newcommand{\vxi}{\vec{\xi}}

\newcommand{\od}{\odot}
\newcommand{\ot}{{\od 2}}

\newcommand{\one}{{\bf 1}}
\newcommand{\blb}{{\boldsymbol\theta}}

\newcommand{\PP}{\mathfrak{P}}

\newcommand{\bpart}{\boldsymbol{\partial}}

\newcommand{\brc}{\mathbf{c}}

\newcommand{\III}{\boldsymbol{\I}}
\begin{document}

\title[A.C. Spectrum on tree-strips of finite cone type]
{Absolutely Continuous Spectrum for random Schr\"odinger operators on tree-strips of finite cone type.}

\author{Christian Sadel}
\address{University of California, Irvine,
Department of Mathematics,
Irvine, CA 92697-3875,  USA}
 \email{csadel@math.uci.edu}


%%%%%%%%%%%%%%%%%%%%%%%%%%%%%%%%%%%%%%%%%%%%%%%%%%%%%%%%%%%%%%%%%%%%%%%%%%%%%%%%%%%%%%%%
%%%%%%%%%%%%%%%%%%%%%%%%%%%%%%%%%%%%%%%%%%%%%%%%%%%%%%%%%%%%%%%%%%%%%%%%%%%%%%%%%%%%%%%%
%%%%%%%%%%%%%%%%%%%%%%%%%%%%%%%%%%%%%%%%%%%%%%%%%%%%%%%%%%%%%%%%%%%%%%%%%%%%%%%%%%%%%%%%
%%%%%%%%%%%%%%%%%%%%%%%%%%%%%%%%%%%%%%%%%%%%%%%%%%%%%%%%%%%%%%%%%%%%%%%%%%%%%%%%%%%%%%%%

\begin{abstract}
A tree-strip of finite cone type is the product of a tree of finite cone type with a finite set.
We consider random Schr\"odinger operators on these tree strips, similar to the Anderson model.
We prove that for small disorder the spectrum is almost surely, purely, absolutely continuous
in a certain set.
\end{abstract}

\maketitle 

\section{Introduction} 

If $\TT$ denotes the set of vertices of a tree, i.e. a discrete graph without loops, then we call the cross product of $\TT$ with a finite set 
$\I=\{1,\ldots,m\}$ a tree-strip. The cardinality $m$ of $\I$ is referred to as 'width' or 'number of orbitals'.
The expressions 'strip' and 'width' come from the fact, that in dimension one, for the 'tree' $\ZZ$, the set $\ZZ\times \I$ corresponds to a strip of width $m$.
The expression 'number of orbitals' refers to the fact that a Schr\"odinger operator on $\TT \times \I$ may model a multi-orbital system 
on $\TT$ since there is a natural isomorphism between the Hilbert spaces $\ell^2(\TT\times \I)$, $\ell^2(\TT)\otimes \CC^m$ and $\bigoplus_{k=1}^m \ell^2(\TT)$. The $m$ copies of $\ell^2(\TT)$ model the orbitals and a Schr\"odinger operator can contain hopping terms along the tree, a potential and interactions between the orbitals.
If $\I=\G$ is a finite graph, then there is a natural adjacency operator on the product graph $\TT\times\G$ and there is an Anderson model on this product graph which can also be considered as a random Schr\"odinger operator on the tree-strip. 

A tree-strip of finite cone type is a tree-strip $\TT\times \I$, where $\TT$ is a tree of finite cone type. Such trees will be constructed below by a substitution rule.
A certain class of operators on trees of finite cone type, including the (ordinary, one-orbital) Anderson model, have been studied in the PhD thesis by Matthias Keller \cite{Kel} and the related papers \cite{KLW, KLW2}. In particular, they showed the existence of purely absolutely continuous spectrum for small disorder for a wide class of trees of finite cone type.

It is widely accepted that for the Anderson model on $\ZZ^d$ and $\RR^d$ and dimension $d\geq 3$ one expects the existence of a.c. (absolutely continuous) spectrum for small disorder and localization for large disorder and at spectral edges. In dimensions $d=1$ and $d=2$ one expects localization for any disorder, except if some built in symmetries prevent localization, e.g. \cite{SS}.

Localization for one-dimensional models \cite{GMP,KuS,CKM}, quasi-one dimensional models (i.e. strips \cite{Lac,KLS}, 
and finite dimensional trees \cite{Breu})
and at spectral edges and high disorder \cite{FS,FMSS,DLS,SW,CKM,DK,Kl2,AM,A,Wang,Klo} is well understood. Localization for low disorder in 2 dimensions and a.c. spectrum for low disorder for $d\geq3$ dimensions remain an open problem.

The existence of a.c. spectrum has only been proved for the Anderson model on trees and tree-like graphs of infinite dimension.
The first proof was done by Klein for regular trees, also called Bethe lattices \cite{Kl3, Kl4, Kl6}. Klein also proved ballistic behavior for the wave spreading on such trees \cite{Kl5}.
Later, different proofs and extensions where given in \cite{ASW,FHS,FHS2,H,AW}. Froese, Hasler, Spitzer and Halasan used hyperbolic geometry and recursive relations of the Green's function in the upper half plane to obtain their results \cite{FHS,FHS2,H}. 
A similar approach is used by Keller, Lenz and Warzel to study the Anderson model on trees of finite cone type \cite{KLW,KLW2}.

A tree-strip where the tree is a regular tree (Bethe lattice) is also called a Bethe strip.
Froese, Hasler and Halasan generalized their method to obtain pure a.c. spectrum for an Anderson model on the Bethe strip of width 2 \cite{FHH}. More precisely, they considered the Anderson model on the product of the regular tree of degree 3 with the finite graph consisting of two vertices and one edge connecting them and proved pure a.c. spectrum in a specific interval. Then, Klein and Sadel extended this result to the Bethe strip with arbitrary degree and arbitrary width \cite{KS}. They used supersymmetric methods as in the original proof by Klein and they could also show ballistic behavior for the wave spreading \cite{KS2}.

All the mentioned results for the existence of a.c. spectrum are valid for small disorder, i.e. small variance of the random potential, and they all rely on some perturbation arguments.
One of the key ingredients in Klein's method is the use of the Implicit Function Theorem which demands to show that 0 is not in the spectrum of a certain Frechet derivative. 
In this work, we will see that this method also works for random Schr\"odinger operators on tree-strips of finite cone type. 
However, the spectrum of the Frechet derivative is given by rather technical expressions leading to a quite technical theorem as one needs to exclude the energies where 0 is in the spectrum of the Frechet derivative. Using some results from \cite{KLW,Kel} we can show that under certain conditions this only excludes a nowhere dense set of energies. For the special case of the (one-orbital) Anderson model on a tree of finite cone type, the result obtained in this article is weaker than the one in \cite{KLW2}. The random potentials treated in \cite{KLW2} are more general
and they only need to exclude finitely many energies for their perturbation arguments.
However, as one can see in \cite{FHH}, the hyperbolic geometry gets a lot more complicated when it is applied to strips. This is the main reason why only a Bethe strip of width 2 has been considered with this method so far.
The supersymmetric method on the contrary does not get more complicated, the width of the strip (number of orbitals) just appears as a parameter in the setting.
The new result in this article compared to the work of Keller, Lenz and Warzel \cite{Kel,KLW2} is the treatment of 'strips', i.e. multi-orbital random Schr\"odinger operators, and the new result compared to our old work \cite{KS} is the treatment of non regular trees of finite cone type.
This leads to some more technicalities in this paper compared to \cite{KS}.
For instance, in order to avoid some extra condition on the type of trees of finite cone type, we work in some slightly different Banach spaces,
$\Hh^{(0)}$ and $\Kk^{(0)}$ as defined in Section~\ref{sec-ban}. Another technical detail that will be dealt with is the fact that
the distribution of the Green's function at different vertices $x$ might be different as one deals with non-regular trees.
As a by product, this work contains the case of a rooted Bethe strip which was not considered in \cite{KS}.
(Note that the number of neighbors at the root for a rooted Bethe lattice is one less than for other vertices.)

\vspace{.2cm}

We will now describe a set of rooted trees of finite cone type that are associated to a substitution matrix $S$, like in \cite{Kel,KLW,KLW2}.
Let $S\in\ZZ_+^{s\times s}$ be an $s\times s$ matrix whose entries are non-negative integers.
Associated to $S$ are the following $s$ rooted trees.
Each vertex $x$ has a label $l(x)=p\in\{1,\ldots,s\}$, a vertex with label $p$ has exactly $S_{p,q}$ children (forward neighbors)
of label $q$, the total number of children is the row sum
$$
S_p:=\sum_{q=1}^s S_{p,q}\;.
$$
Hence, except for the root, a vertex of label $p$ has $S_p+1$ neighbors (one parent and $S_p$ children) and the root has
$S_q$ neighbors if it has label $q$.
Such an infinite tree is uniquely determined (up to tree isomorphisms) by $S$ and the label of the root.
We denote the tree where the root has label $q$ by $\IS^{(q)}$.
If one cuts the path going to the root at a vertex with label $p$, then one obtains the tree $\IS^{(p)}$, or in other words,
the cone of descendants at each vertex of label $p$ is isomorphic to $\IS^{(p)}$. Hence, there are only finitely many cones and therefore the tree
is said to be of finite cone type. Vice versa, each tree of finite cone type can be constructed this way.
Figure~\ref{tree} shows an example of a tree of finite cone type associated to $\left(\begin{smallmatrix} 2&1 \\2&2\end{smallmatrix}\right)$.

\newsavebox{\smlmat}% Box to store smallmatrix content to use it in caption without producing error.
\savebox{\smlmat}{$\left(\begin{smallmatrix} 2&1 \\2&2\end{smallmatrix}\right)$}

\begin{figure}[ht]
\begin{center}
 \includegraphics[width=7cm]{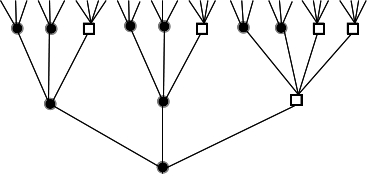}
\end{center}
\caption{Tree $\TT^{(1)}$ associated to substitution matrix 
\usebox{\smlmat}. Vertices of label 1 are filled circles and vertices of label 2 are non-filled squares.}
\label{tree}
\end{figure}

As one can think of a vertex of label $p$ as being substituted by
$S_p$ vertices in the next generation, where $S_{p,q}$ of them are of label $q$,
the matrix $S$ is called the 'substitution matrix'. Trees of finite cone type may also be called 'substitution trees' for this reason.
The special case of a rooted Bethe lattice with connectivity $K$ ($K$ children for each vertex) is given by the $1\times 1$ substitution matrix $S=K$. The regular tree of degree $K+1$ is given by $\IS^{(1)}$ 
using the substitution matrix $S=\left(\begin{smallmatrix} 0 & K+1 \\ 0 & K \end{smallmatrix}\right)$.

Instead of considering these trees individually, we will simply consider the forest $\IS=\bigcup_{q=1}^s \IS^{(q)}$ consisting of the $s$ connected trees $\IS^{(q)}$.

For $x,y\in\IS^{(q)}$ let $d(x,y)$ denote the distance, i.e. the length of the shortest path from $x$ to $y$.
If $x\in\IS^{(p)}$ and $y\in\IS^{(q)}$ with $p\neq q$, then let $d(x,y)=\infty$.
We will consider the self-adjoint operators
\begin{equation} \label{eq-H_lb}
(H_{\lb} u)(x)=\left(\sum_{y:d(x,y)=1} u(y)\right) + Au(x) + \lb V(x) u(x)
\end{equation}
on the Hilbert space of $\CC^m$ valued $\ell^2$ functions on $\IS$, $\ell^2(\IS,\CC^m)\ni u$.
This means $u(x)\in\CC^m$ and $\sum_{x\in\IS} \|u(x)\|^2 < \infty$.
The Hilbert space $\ell^2(\IS,\CC^m)$ is canonically equivalent to $\bigoplus_{q=1}^s \ell^2(\IS^{(q)},\CC^m) \cong \ell^2(\IS)\otimes \CC^m \cong
\ell^2(\IS\times\{1,\ldots,m\})$. The set of real symmetric $m\times m$ matrices will be denoted by $\Sym(m)$ and
$A\in \Sym(m)$ represents the 'free vertical operator'.
The matrices $V(x)\in\Sym(m)$ for $x\in\IS$ are independent identically distributed random variables, distributed according to the
probability measure $\nu$ on $\Sym(m)$.
$A$ and $V(x)$ describe the potential and the interactions between the $m$ orbitals.
Clearly, $H_\lb=\bigoplus_{q=1}^s H^{(q)}_{\lb}$ where $H^{(q)}_{\lb}$ is the restriction of $H_\lb$ to $\ell^2(\IS^{(q)},\CC^m)$ and can be seen
as  random Schr\"odinger operator on the tree strip $\IS^{(q)}\times \I$.

For $\lambda=0$ one has $H_0=\Delta \otimes \one + \one \otimes A$ on $\ell^2(\IS)\otimes\CC^m$ where $\Delta$
describes the adjacency operator on $\ell^2(\IS)$ given by
\begin{equation}
 (\Delta v)(x)=\sum_{y:d(x,y)=1} v(y)\;,\qquad v\in\ell^2(\IS)\;.
\end{equation}
Setting $A$ to be the adjacency matrix for a finite graph $\G$
and $\nu$ to be supported on the diagonal matrices, i.i.d. in each diagonal entry, we obtain 
the Anderson model on the product of the finite graph $\G$ with $\IS$. Setting $A=0$ and $\nu$ to be the
distribution as in the orthogonal ensemble (GOE), we obtain the Wegner $m$-orbital model on the forest $\IS$.

Let us remark that there is an orthogonal matrix $O\in{\rm O}(m)$ such that $O^\top A O$ is diagonal.
Then, using the equivalence $\ell^2(\IS,\CC^m)\cong\ell^2(\IS)\otimes\CC^m$, the operator $(\one\otimes O)$ is unitary and one obtains
\begin{equation}
 \left[ (\one\otimes O)^* H_\lambda (\one\otimes O) \right] u(x)=
\left(\sum_{y:d(x,y)=1} u(y)\right) + O^\top A O u(x) + \lambda O^\top V(x) O u(x)\;.\notag
\end{equation}
Hence, without loss of generality, we can assume that $A$ is a diagonal matrix and we will do so in the proofs.
In particular, the non-random operator $H_0$ is unitarily equivalent to a direct sum of shifted adjacency operators on $\TT$,
$H_0\cong \bigoplus_{j=1}^m \Delta+a_j$, where the $a_j$ are the eigenvalues of $A$.

Our interest lies in the spectral type of $H_\lb$ which is determined by the matrix-valued spectral measures at the 
vertices of the trees $\IS^{(q)}$. For $x\in\IS,\,j\in\{1,\ldots,m\}$
let $|x,j\rangle$ denote the
element in $\ell^2(\IS,\CC^m)$ satisfying $|x,j\rangle (y) = \delta_{x,y} e_j$ where $e_j$ is the $j$-th canonical basis vector in $\CC^m$.
Moreover, for an operator $H$ we denote by $\langle x,j| H |y,k\rangle$ the scalar product between
$|x,j\rangle$ and $H|y,k\rangle$ with the convention that the scalar product is linear in the second and anti-linear in the first component.
Then, for $x\in \IS$ we define the random, positive matrix valued measure $\mu_x$ on $\RR$ by
\begin{equation}
 \int f(E)\,d\mu_x = \left[ \,\langle x,j |f(H_\lambda)|x,k\rangle\, \right]_{j,k\in\I}\;
\end{equation}
for all compactly supported, continuous functions $f$ on $\RR$.

The roots of $\IS^{(q)}\subset\IS$ will have a special role, therefore let us denote them by $0_q\in\IS^{(q)}$.
Similarly to above, we let $|x\rangle$ denote the function $|x\rangle(y)=\delta_{x,y}$ on $\ell^2(\IS)$ and for
$\im(z)>0$ we define
\begin{equation} \label{eq-def-Gamma}
 \Gamma^{(q)}_z := \langle 0_q | (\Delta-z)^{-1} | 0_q \rangle.
\end{equation}
Furthermore, we define 
\begin{equation}
\Gamma^{(q)}_E:=\lim_{\eta\downarrow 0} \Gamma^{(q)}_{E+i\eta}
\end{equation}
for real energies $E$, if this limit exists as a limit in $\CC$ (if the limit approaches infinity it is considered as non-existent).
Let $\Delta_q$ denote the adjacency operator on $\TT^{(q)}$ which is the restriction of $\Delta$ to $\TT^{(q)}$ and
define the set $I_q \subset \sigma(\Delta_q)$ by 
\begin{equation} \label{eq-def-Iq}
 I_q:=\left\{E\in\RR\,:\, \Gamma^{(q)}_E=\lim_{\eta\downarrow 0} \Gamma^{(q)}_{E+i\eta}\;\; \text{exists in $\CC$ and}\;\;
\im(\Gamma^{(q)}_E)>0\;\right\}.
\end{equation}
Now let $a_1\leq a_2\leq\ldots\leq a_m$ be the eigenvalues of the free vertical operator $A$. 
Then, we define
\begin{equation} \label{eq-def-hat-I}
 I_S:=\bigcap_{q=1}^s I_q\;,\quad  I_{A,S}:= \bigcap_{j=1}^m (I_S+a_j) =\{E: E-a_j\in I_S \;,\;\forall\,j=1,\ldots,m\}\;.
\end{equation}
Note that $H_0\cong\bigoplus_{j=1}^m \Delta+a_j \cong \bigoplus_{j=1}^m \bigoplus_{q=1}^s \Delta_q + a_j$.
The set $I_{A,S}$ is the set of energies, such that for all $q$ and $a_j$ the limit 
$\lim_{\eta\downarrow 0} \langle 0_q|[(\Delta_q+a_j)-(E+i\eta)]^{-1} |0_q \rangle$ exists and the imaginary part of it is positive.
In particular, 
\begin{equation}
I_{A,S}\subset\bigcap_{j=1}^m \bigcap_{q=1}^s \sigma(\Delta_q+a_j) \subset \sigma(H_0) \;.
\end{equation}

Let further $\Delta(m,\ZZ_+)$ denote the set of upper triangular matrices with non-negative integer entries.
For $J\in\Delta(m,\ZZ_+)$ and $\im(z)>0$ as well as for $z=E\in  I_{A,S}$ we define
\begin{align}
\theta^{(q)}_{J,z} &:= \prod_{\substack{j,k\in\{1,\ldots,m\} \\j\leq k}} \left[\Gamma^{(q)}_{z-a_j} \Gamma^{(q)}_{z-a_k} \right]^{J_{jk}}\;\in\,\CC\qtx{and} \\
\blb_{J,z}&:=\diag(\theta^{(1)}_{J,z},\theta^{(2)}_{J,z},\ldots,\theta^{(s)}_{J,z})\;\in\;\CC^{s\times s}.
\label{eq-def-blb}
\end{align}
By $\diag(m_1,\ldots, m_s)$ we denote the diagonal matrix with entries $m_1,\ldots, m_s$ on the diagonal.
Furthermore, let us define
\begin{equation}
 |J|:=\sum_{j,k=1}^m J_{j,k}\,\in\ZZ_+\quad\text{for}\quad J\in\Delta(m,\ZZ_+)\;.
\end{equation}

The method of \cite{KS} will be applicable to the set of energies $E\in  I_{A,S}$ where
\begin{equation}\label{eq-cond-blb}
\det(\blb_{J,E}\blb^*_{J',E}S-\one)\neq 0\qtx{for all} J,J' \in \Delta(m,\ZZ_+) \qtx{with} |J|+|J'|\geq 1\;.
\end{equation}
Hence, let us define the set
\begin{equation} \label{eq-def-I_AS}
 \hat I_{A,S}:=  \bigcap_{\substack{J,J'\in\Delta(m,\ZZ_+)\\ |J|+|J'|\geq 1}}
 \big\{E\in I_{A,S}\,:\, \det(\blb_{J,E} \blb^*_{J',E} S-\one) \neq 0 \,\big\} \,\subset\,I_{A,S}.
\end{equation}
This means if $E\in \hat I_{A,S}$, then $1$ is not an eigenvalue of $\blb_{J,E} \blb^*_{J',E} S$ for $|J|+|J'|\geq1$.
The set $\hat I_{A,S}$ will be precisely the set where we will be able to use the Implicit Function Theorem, similar as in \cite{KS}.

\vspace{.2cm}

\begin{assumptions} 
The following assumptions on the distribution on the potential $V(x)$ and on the substitution matrix $S$ will play an important role.
\noindent \begin{enumerate}
\item[{\rm (V)}] We assume that all mixed finite moments of the random entries of the random matrix $V(x)$ exist.
This implies that partial derivatives of the Fourier transform
\begin{equation} \label{eq-def-h}
 h(M):=\E \left(e^{-i\Tr(MV(x))}\right)\,,\quad M\in\Sym(m)
\end{equation}
exist to any order and are bounded.
\item[{\rm (S1)}] $S_p=\sum_q S_{p,q} \geq 2$, for any $p\in\{1,\ldots,s\}$, i.e. each vertex has at least 2 children.
\item[{\rm (S2)}] For all $p,q\in\{1,\ldots,s\}$ there exists a natural number $n$ such that the matrix entry $(S^n)_{p,q}$ of $S^n$ is positive. This  means that for all $p,q\in\{1,\ldots,s\}$ the tree $\IS^{(p)}$ contains vertices labeled by $q$.
\item[{\rm (S3)}] $\|S\|<K^2$ for $K:=\min\big\{S_{q,q}\,:\, q\in\{1,\ldots,s\}\big\}$.
\end{enumerate}
\end{assumptions}

\begin{remark}
\noindent \begin{enumerate}[{\rm (i)}]
\item For $S\neq 0,\,S\in\ZZ_+^{s\times s}$, assumption {\rm (S3)} implies $S_{q,q}\geq K\geq 2$ since $\|S\|\geq 1$.
This in turn implies {\rm (S1)}.
\item Assumptions {\rm (V)} and {\rm (S1)} are important to be able to use the method from \cite{KS} for energies $E\in \hat I_{A,S}$
(cf. Theorem~\ref{main} below).
Assumptions {\rm (S2)} and {\rm (S3)} will assure that the set $\hat I_{A,S}$ is a dense open subset of the interior
$ \mathring{I}_{A,S}$ of $I_{A,S}$ and not empty for $\|A\|$ small enough (cf. Theorem~\ref{main1}~(i)). 
Therefore, Theorem~ \ref{main} is not an empty statement.
\end{enumerate}
\end{remark}

\begin{theorem} \label{main}
If assumptions {\rm (V)} and {\rm (S1)} are satisfied, then
there is an open neighborhood $U$ of $\{0\}\times \hat I_{A,S}$ in $\RR^2$ such that for $U_\lambda=\{E:(\lambda,E)\in U\}$ one has the following:
\begin{enumerate}[{\rm (i)}]
\item 
The spectrum of $H_\lambda$ is almost surely
purely absolutely continuous in $U_\lambda$.
\item For every $x\in\IS$ the average spectral measure $\E(\mu_{x})$ is absolutely continuous with respect to the Lebesgue measure in 
$U_\lambda$ and the density is a positive semi-definite matrix valued function which depends continuously on $(\lambda,E)\in U$.
Moreover, at the roots $0_q$, the density of $\E(\mu(0_q))$, $q=1,\ldots,s$, is a positive definite matrix valued function in $U_\lb$, showing that
there is spectrum in $U_\lambda$ with positive probability.
\end{enumerate}
\end{theorem}

In order to show that this is not an empty statement in the sense that the set $\hat I_{A,S}$ is not always empty, 
we will also show the following.

\begin{theorem} \label{main1}
Let the substitution matrix $S$ satisfy {\rm (S2)} and {\rm (S3)}. Then, the following hold:
\begin{enumerate}[{\rm (i)}]
\item The interior of the set $I_S$ is not empty and consists of finitely many intervals. 
In particular, letting $a(S)$ denote the length of the longest of these intervals in $I_S$ and letting $a_{{\rm max}}$ and $a_{{\rm min}}$
be the largest and smallest eigenvalue of $A$, one obtains:
If $a_{{\rm max}}-a_{{\rm min}} < a(S)$, then the interior $\mathring{I}_{A,S}$ of the set $ I_{A,S}$ 
is not empty and consists of finitely many intervals.
\item $\hat I_{A,S}$ is a dense open set in $\mathring{I}_{A,S}$, i.e. the closure of
$\hat I_{A,S}$ contains $\mathring{I}_{A,S}$.
Consequently, if $a_{{\rm max}}-a_{{\rm min}} < a(S)$ and $\lambda$ is small enough, then the set 
$U_\lambda$ as in Theorem~\ref{main} contains some intervals and the theorem is not an empty statement.
\item For a natural number $b\in\NN$ the matrix $bS$ also satisfies
{\rm (S2)} and {\rm (S3)} and one obtains $I_{bS}=\sqrt{b}\,I_S$ and hence $a(bS)=\sqrt{b} \,a(S)$.
In particular, for fixed $A$ and $b$ large enough one has $a_{{\rm max}}-a_{{\rm  min}}<a(bS)$ and $\hat I_{A,bS}$ is not empty.
\end{enumerate}
\end{theorem}

\begin{remark}
\noindent \begin{enumerate}[{\rm (i)}]
\item The meaning of assumption {\rm (S1)} is that it forbids any kind of line segments.
In particular, non of the trees $\TT^{(p)}$ can be isomorphic to the lattice of positive integers $\ZZ_+$.
But it also forbids the trees where the usual Anderson model (not strips) was treated in
the PhD thesis by Halasan \cite{H}, such as the Fibonacci trees associated to the substitution matrix
$S=\left( \begin{smallmatrix} 0 & 1 \\ 1 & 1 
\end{smallmatrix} \right)$.
If we denote by $\#_n(\TT^{(p)})$ the total number of vertices in the $n$-th generation (with the root being the 1st generation), 
then the Fibonacci trees satisfy 
$\#_n(\TT^{(1)})=f_n, \#_n(\TT^{(2)})=f_{n+1}$ where $(f_n)_n$ is the Fibonacci sequence starting with $f_1=f_2=1$. 

With some technical adjustments one can treat the Fibonacci tree-strip as well.
However, these adjustments are quite different from the ones needed in this paper. The Fibonacci tree-strip is a very special case 
where {\rm (S1)} is not satisfied and it will be dealt with elsewhere.

The real necessary assumption should be that no tree $\TT^{(p)}$ is isomorphic to $\ZZ_+$. This case needs to be excluded as the Anderson model on $\ZZ_+$ leads to localization even for small disorder.

\item 
One way to satisfy assumption {\rm (S2)} is to make all entries in the secondary diagonals  (above and below the diagonal) positive. 
As the Hilbert-Schmidt norm is an upper bound for the usual matrix norm, {\rm (S3)}
is satisfied if $\sum_{p,q} (S_{p,q})^2 < K^4$. 
These facts lead to a class of substitution matrices satisfying {\rm (S2)} and {\rm (S3)},
such as e.g. $\left(\begin{smallmatrix} 2 & 1 \\ 2 & 2 \end{smallmatrix} \right)$ or
$\left(\begin{smallmatrix} 4 & 1 & 0\\ 1 & 3 & 2 \\ 0 & 3 & 3 \end{smallmatrix} \right)$.
The tree $\TT^{(1)}$ associated to 
$S=\left(\begin{smallmatrix} 2 & 1 \\ 2 & 2 \end{smallmatrix} \right)$  is given in figure~\ref{tree}.

\item I think that one can improve Theorem~\ref{main1}~(ii) and conjecture that the set $I_{A,S}\setminus \hat I_{A,S}$ is always finite.
In fact, for rooted regular trees (Bethe lattices) of degree $K\geq 2$, the set $I_{A,S}\setminus\hat I_{A,S}$ is actually empty
(cf. \cite{KS}).
The statement that $\hat I_{A,S}$ is a dense open subset of $\mathring {I}_{A,S}$ is quite a lot weaker. For instance, 
$ \mathring{I}_{A,S} \setminus \hat I_{A,S}$ might be a cantor set consisting of uncountably many points.

 \item Part (i) and (iii) of Theorem~\ref{main1} are still valid if one replaces assumption {\rm (S3)} by the weaker one:\\
{\rm (S3')}\;\; {\it $S_{q,q}\geq 1$ for all $q\in\{1,\ldots,s\}$}\,.\\
In fact, part (i) is already proved in \cite{Kel,KLW2} and the assumptions 
{\rm (S1)}, {\rm (S2)} and {\rm (S3')} together are equivalent to assumptions {\rm (M0), (M1)} and {\rm (M2)} in \cite{Kel,KLW2}.
In view of my conjecture above, I also expect part (ii) to  be true in this case.
An example of a substitution matrix satisfying {\rm (S1)}, {\rm (S2)} and {\rm (S3')} but not {\rm (S3)} is $S=\left(\begin{smallmatrix}  1 & 1 \\ 1 & 2
\end{smallmatrix} \right)$. The trees associated to this matrix can be obtained from the Fibonacci-trees by removing every 2nd generation of vertices. The number of vertices in the $n$-th generations are $\#_n(\TT^{(1)})=f_{2n-1}$ and $\#_n(\TT^{(2)})=f_{2n}$, where $f_n$ is the $n$-th Fibonacci number.

The substitution matrix for the Fibonacci-tree as given above
is an interesting example satisfying {\rm (S2)} but not {\rm (S3')} and also not {\rm (S1)}.

\item Part (iii) of Theorem~\ref{main1} is particularly interesting for the Anderson model on a product graph where $A$ is fixed to be the adjacency matrix of a  finite graph $\G$. It shows that there are substitution matrices $S$ (and corresponding trees) 
such that the set $\hat I_{A,S}$ is not empty.
\item  I expect the set $\hat I_{A,S}$ to be non empty in many more cases than the once covered by Theorem~\ref{main1}. 
However, since one can not obtain explicit formulas for the Greens functions $\Gamma_z^{(q)}$ as defined in 
\eqref{eq-def-Gamma} in general, it is not so simple to show that the set $\hat I_{A,S}$ is in fact not empty.
\end{enumerate}
\end{remark}

The important objects we work with are the matrix Green's functions given by
\begin{equation}\label{eq-def-G^x}
 G_\lambda^{[x]}(z):=\left[\, \langle x,j|(H-z)^{-1} | x,k\rangle\,\right]_{j,k\in\I}\,\in\,{\CC^{m\times m}}
\end{equation}
for $\im(z)>0$.
The most important ingredient to obtain Theorem~\ref{main} is the following.

\begin{theorem}\label{main2}
 Under assumptions {\rm (V)} and {\rm (S1)} there exists an open neighborhood $U$ of $\{0\}\times \hat I_{A,S}$ in $\RR^2$ 
such that for all vertices $x\in\IS$ the functions
\begin{align} 
(\lambda,E,\eta) &\mapsto \E\left(G_\lambda^{[x]}(E+i\eta)\right)\;, \notag \\
(\lambda,E,\eta) &\mapsto \E\left(\left|G_\lambda^{[x]}(E+i\eta)\right|^2\right)\;, \notag
\end{align}
defined for $\eta>0$, have continuous extensions to $U\times [0,\infty)$.
\end{theorem}

%\vspace{.2cm}

\vspace{.2cm}

We will first prove Theorem~\ref{main1} in Section~\ref{sec-proofmain1}. 
Then, in Section~\ref{sec-ban}, we introduce the important Banach spaces that were also used in \cite{KS}.
Appendix~\ref{sec-super} will give the super-symmetric formalism that leads to these spaces.
In Section~\ref{sec-fxp} we derive some fixed point equations in these Banach spaces. Next,  we calculate the Frechet derivative of the operators appearing in these fixed point equations in Section~\ref{sec-frechet}.
Finally, in Section~\ref{sec-proofs} we use the Implicit Function Theorem to obtain Theorem~\ref{main2} and Theorem~\ref{main}.

\vspace{.2cm}

\noindent {\bf Acknowledgement.} I am thankful to Matthias Keller for interesting discussions.

%%%%%%%%%%%%%%%%%%%%%%%%%%%%%%%%%%%%%%%%%%%%%%%%%%%%%%%%%%%%%%%%%%%%%%%%%%%%%%%%%%%%%%%%%%%%%%%%%%%%%%%%%%5
%%%%%%%%%%%%%%%%%%%%%%%%%%%%%%%%%%%%%%%%%%%%%%%%%%%%%%%%%%%%%%%%%%%%%%%%%%%%%%%%%%%%%%%%%%%%%%%%%%%%%%%%%%5
%%%%%%%%%%%%%%%%%%%%%%%%%%%%%%%%%%%%%%%%%%%%%%%%%%%%%%%%%%%%%%%%%%%%%%%%%%%%%%%%%%%%%%%%%%%%%%%%%%%%%%%%%%5
%%%%%%%%%%%%%%%%%%%%%%%%%%%%%%%%%%%%%%%%%%%%%%%%%%%%%%%%%%%%%%%%%%%%%%%%%%%%%%%%%%%%%%%%%%%%%%%%%%%%%%%%%%5

\section{Proof of Theorem~\ref{main1}} \label{sec-proofmain1}

%%%%%%%%%%%%%%%%%%%%%%%%%%%%%%%%%%%%%%%%%%%%%%%%%%%%%%%%%%%%%%%%%%%%%%%%%%%%%%%%%%%%%%%%%%%%%%%%%%%%%%%%%%5
%%%%%%%%%%%%%%%%%%%%%%%%%%%%%%%%%%%%%%%%%%%%%%%%%%%%%%%%%%%%%%%%%%%%%%%%%%%%%%%%%%%%%%%%%%%%%%%%%%%%%%%%%%5

The following observations are important.
As in \cite{KLW,Kel} define $\Oo$ to be the system of all open sets $O$ in $\RR$, such that 
all the Green's functions $z\mapsto \Gamma^{(q)}_z$, 
defined on the upper half plane $\HH=\{z\in\CC:\im(z)>0\}$, extend continuously to $O\cup\HH$
with $\im(\Gamma^{(q)}_E)>0$ for $E\in O$.
Then let
\begin{equation}\label{eq-def-Sigma}
\Sigma_S:=\bigcup_{O\in\Oo} O\;.
\end{equation}
Clearly, $\Sigma_S$ is the largest set in $\Oo$ and the Green's functions as mentioned above extend continuously to $\Sigma_S \cup \HH$. Moreover, $\Sigma_S\subset I_q$, by definition of $I_q$.
The following lemma is a consequence of Theorem~6 in \cite{KLW}.
\begin{lemma}\label{lem-1}
Let $S$ satisfy {\rm (S2)} and {\rm (S3)}. Then, the following holds:
\begin{enumerate}[{\rm (i)}]
\item $\Sigma_S\subset I_S = \bigcap_{q=1}^s I_q$ 
\item $\Sigma_S$ consists of finitely many intervals.
\item The closure of $\Sigma_S$ is equal to the spectrum of the adjacency operator $\Delta$ on $\IS$, i.e.
$\overline \Sigma_S  =\sigma(\Delta)\supset \bigcup_{q=1}^s I_q$.
\item $\sigma(\Delta)\setminus \Sigma_S$ is finite, and hence $I_q \setminus \Sigma_S$ is finite for every 
label $q\in\{1,\ldots,s\}$.
\end{enumerate}
\end{lemma}

One also obtains the following.

\begin{lemma}\label{lem-2}
Let $S$ satisfy {\rm (S2)} and {\rm (S3)} and let 
us denote the dependence of the Green's functions $\Gamma^{(q)}_z$ for the adjacency operator $\Delta$
on the substitution matrix $S$ (which defines $\IS$) by $\Gamma^{(q)}_z(S)$. 
Then, $bS$ satisfies {\rm (S2)} and {\rm (S3)} as well and for a natural number $b\in\NN$ one finds
\begin{equation}\label{eq-gm-bS-S}
\Gamma^{(q)}_{\sqrt{b} z}(bS) = \frac1{\sqrt{b}}\,\Gamma^{(q)}_z (S)\qtx{for} \im(z)>0\;.
\end{equation}
\end{lemma}

\begin{proof}
To see that $bS$ satisfies {\rm (S2)} is straight forward, for {\rm (S3)} note that with
$K=\min\{S_{q,q}:q=1,\ldots,s\},\;\|S\|<K^2$ one has $(bS)_{q,q}\geq bK$ and $\|bS\|=b\|S\|<bK^2<(bK)^2$.

As shown in \cite{KLW,Kel}, under the assumptions {\rm (S2)} and {\rm (S3)} for $\im(z)>0$ the Green's functions are uniquely
determined by the equations
\begin{align}
 \Gamma_z^{(p)}(S) \left( z+\sum_{q=1}^s S_{p,q} \Gamma_z^{(q)} (S) \right) +1&=0\;, \quad \im\left(\Gamma_z^{(q)}(S)\right)>0\;; \label{eq-gm-S}\\
 \Gamma_z^{(p)}(bS) \left( z+\sum_{q=1}^s bS_{p,q}\Gamma_z^{(q)} (bS) \right) +1&=0\;, \quad \im\left(\Gamma_z^{(q)}(bS)\right)>0 \label{eq-gm-bS}\;.
\end{align}
If $\Gamma_z^{(q)}(S)$ satisfy \eqref{eq-gm-S} then it is a simple calculation to verify that
$\Gamma^{(q)}_z(bS)$ defined by \eqref{eq-gm-bS-S} satisfy \eqref{eq-gm-bS}.
\end{proof}

Now we can go ahead with the proof of the theorem.

\renewcommand{\proofname}{Proof of Theorem~\ref{main1}}
\begin{proof}
Let $S$ satisfy  {\rm (S2)} and {\rm (S3)}.
Let $\Sigma_S$ be defined as in \eqref{eq-def-Sigma}. Then, Lemma~\ref{lem-1} states that $\Sigma_S\subset I_S$,
$I_S\setminus \Sigma_S$ is finite and $\Sigma_S$ consists of finitely many open intervals. As $\bar \Sigma_S = \sigma(\Delta)\neq \emptyset$ 
we also find that the interior of $I_S$ is not empty and consists of finitely many intervals. 
Let $a(S)$ be the length of the largest of those intervals. Then, the definition \eqref{eq-def-hat-I} of $ I_{A,S}$ yields that $ I_{A,S}$ is not empty for $a_{{\rm max}}-a_{{\rm min}} <a(S)$. This shows part (i).

For part(ii) let us define
\begin{equation}
\Sigma_{A,S}:=\bigcap_{j=1}^m (\Sigma_S+a_j)=\{E:E-a_j\in\Sigma\,,\;\forall\,j=1,\dots,m\}\;.
\end{equation}
Then, Lemma~\ref{lem-1} implies that $\Sigma_{A,S}$ is an open subset of $ I_{A,S}$ and hence also of its interior $\mathring{I}_{A,S}$. Furthermore, $\Sigma_{A,S}$ consists of finitely many intervals and 
$\mathring{I}_{A,S} \setminus \Sigma_{A,S}$ is finite. Hence, $\Sigma_{A,S}$ is a dense open set in $\mathring{I}_{A,S}$.

As the dependence of $I_{A,S}$ on $A$ lies only in the eigenvalues of $A$, one may assume $A$ to be diagonal. Then, for $\lambda=0$ the matrix Green's function $G_0^{(q)}(z)$ is diagonal and one has
$G_0^{(q)}(z)=\diag(\Gamma^{(q)}_{z-a_1},\ldots,\Gamma^{(q)}_{z-a_m})$.
Therefore, the continuity statement of Theorem~\ref{main2} for the case $\lambda=0$ implies that 
$\hat I_{A,S}$ lies inside $\Sigma_{A,S}$, so $\hat I_{A,S}\subset \Sigma_{A,S}$. 

We now prove that $\hat I_{A,S}$ is a dense open subset of $\Sigma_{A,S}$.
Since $\Sigma_{A,S}$ is a dense open set in $\mathring{I}_{A,S}$, this implies that
$\hat I_{A,S}$ is a dense open set in $\mathring{I}_{A,S}$.

Let $K=\min\big\{S_{q,q}:q\in\{1,\ldots,s\}\big\}$, then one has $K\geq 2$ by {\rm (S3)}. 
Then, \cite[Lemma~3]{KLW} or alternatively, \cite[Lemma~3.2]{Kel} gives
$|\Gamma^{(q)}_{E-a_j}|\leq\frac{1}{\sqrt{S_{q,q}}}\leq\frac1{\sqrt{K}} $ for all $q$.
This implies
\begin{equation}
 \|\blb_{J,E} \|\leq K^{-|J|}\quad\text{and}\quad
 \left\| \blb_{J,E} \blb^*_{J',E} S \right\| < K^{2-|J|-|J'|}\;,
\end{equation}
where we used $\|S\|<K^2$ as stated in assumption {\rm (S3)}. Thus,
\begin{equation} \label{eq-J+J'>=2}
\det(\blb_{J,E} \blb^*_{J',E} S-\one)\neq 0\quad\text{for}\quad
|J|+|J'|\geq 2\;.
\end{equation}
Therefore, the only possibility to have $E\in \Sigma_{A,S}\setminus \hat I_{A,S}$ is if $\det(\blb_{J,E} \blb^*_{J',E} S-\one)=0$
for some $J,J'\in\Delta(m,\ZZ_+)$ with $|J|+|J'|=1$, i.e. either $J=0$ and $|J'|=1$ or
$|J|=1$ and $J'=0$. 
Let us consider the latter case first and define $f_J(z):=\det(\blb_{J,z} S-\one)$
for $\im(z)>0$ and $z=E\in \Sigma_{A,S}$.
Moreover, let $N_J=\{E\in \Sigma_{A,S}: f_J(E)=0\}$.

By definition of $\Sigma_{A,S}$, the functions $z\mapsto \blb_{J,z}$ extend continuously to $\Sigma_{A,S}\cup\HH$.
Therefore, by continuity the set $N_J \cap \Sigma_{A,S}$ is closed in $\Sigma_{A,S}$ (with respect to the relative topology
in $\Sigma_{A,S}$).

Claim: $N_J \subset  \Sigma_{A,S}$ is a closed, nowhere dense set in $\Sigma_{A,S}$ (i.e. the interior w.r.t. the topology in 
$\Sigma_{A,S}$ is empty).

Assume the interior of $N_J$ in $\Sigma_{A,S}$ is not empty. This means, there is some non-empty open interval $I\subset N_J$. Then,
$f_J(E)=0$ for $E\in I$, so $f_J$ restricted to $I$ is real and the limit
of the holomorphic function $f_J(z)$ for $\im(z)>0$ with $z\to E$. Therefore, by the Schwarz reflection principle, $f$ extends to a 
holomorphic function in a neighborhood of $I$ in the complex plane by defining $f(z)=\bar f(\bar z)$ for $\im(z)<0$.
But since $f_J$ restricted to $I$ is zero this means $f$ would be identically zero in a complex neighborhood of $I$ which
contains an open set in the upper half plane. Therefore $f(z)$ would be zero in the entire upper half plane.
However, as the Green's functions $\Gamma^{(q)}_z$ go to zero if the imaginary part of $z$ goes to infinity, one obtains
\begin{equation}
 \lim_{\eta\to\infty} \blb_{J,E+i\eta} = 0\quad\text{implying}\quad
\lim_{\eta\to\infty} f_J(E+i\eta)= \det(-\one)=(-1)^s\;
\end{equation}
which gives a contradiction. Therefore, the interior of $N_J$ is empty and $N_J$ is a closed, nowhere dense set in $\Sigma_{A,S}$.

Using the anti-holomorphic function $g_J(z)=\det(\blb_{J,z}^*S-\one)$ for $\im(z)>0$, the same arguments give that the set
$\tilde N_J=\{E\in \Sigma_{A,S}:g_J(E)=0\}$ is also closed and nowhere dense in $ \Sigma_{A,S}$.
Now, finite unions of closed, nowhere dense sets are closed and nowhere dense and 
complements of closed nowhere dense sets are open dense sets.
Therefore, by  \eqref{eq-def-I_AS} and \eqref{eq-J+J'>=2} one obtains that
\begin{equation}
 \hat I_{A,S} =  \Sigma_{A,S} \setminus \Big( \bigcup_{\substack{J\in\Delta(m,\ZZ_+)\\ |J|=1}} N_J \cup \tilde N_J \Big)
\end{equation}
is a dense open set in $ \Sigma_{A,S}$.
This finishes the proof of part (ii). 
%%%%%

Part (iii) basically follows from Lemma~\ref{lem-2}. In particular,
for the set $I_S$ as defined in \eqref{eq-def-hat-I}, equation \eqref{eq-gm-bS-S} 
implies $I_{bS}=\sqrt{b}\,I_S$ which by the definition of $a(S)$ above implies
$a(bS)=\sqrt{b}\,a(S)$.
\end{proof}
\renewcommand{\proofname}{Proof} 

\begin{remark} 
\noindent \begin{enumerate}[{\rm (i)}]
\item A similar argument can be used to obtain that the limits 
$\Gamma^{(q)}_E=\lim\limits_{\eta\downarrow 0}\Gamma^{(q)}_{E+i\eta} $ depend analytically on $E$ for a dense open set in $\Sigma_S$.
By \eqref{eq-gm-S}, $\Gamma^{(q)}_E$ is implicitly defined for $z=E\in \Sigma_S$ and for $\im(z)>0$ by 
$F(z,\Gamma^{(1)}_z,\ldots,\Gamma^{(s)}_z)=0$ where $F=(F_1,\ldots,F_s)$ with 
$$F_p(z,x_1,\ldots,x_s)=x_p\,\left(z+\sum_{q=1}^s S_{p,q} x_q\right)+1\;.$$
Then $\partial_{x_r} F_p = \delta_{p,r} (z+\sum_q S_{p,q} x_q) + x_p S_{p,r}\;.$
Using \eqref{eq-gm-S} and defining the diagonal matrix $\Gamma_z=\diag(\Gamma^{(1)}_z,\ldots, \Gamma^{(s)}_z)$ this leads to the Jacobi matrix
$$
\left[\partial_{x_r} F_p(z,\Gamma^{(1)}_z,\ldots,\Gamma^{(s)}_z)\right]_{p,r=1,\ldots,s} =
-\Gamma_z^{-1} + \Gamma_z S = \Gamma_z^{-1}(\Gamma_z^2 S - \one)\;.
$$
By the same arguments as used above, the determinant of this Jacobi matrix is not zero in a dense open subset $\tilde I$ in $I_q$. 
Using the analytic version of the Implicit Function Theorem, this implies that $\Gamma_E^{(q)}$ is analytic in $E$ for $E\in\tilde I$.
\item For $J(k) \in \Delta(m,\ZZ_+)$ defined by $|J(k)|=1, [J(k)]_{k,k}=1$ one finds 
$\blb_{J(k),E}=\Gamma^2_{E-a_k}$. As $\det(\Gamma_{E-a_k}^2 S-\one)=\det(\blb_{J(k),E} S -\one)\neq 0$ for $E\in \hat I_{A,S}$
the arguments above give that $\blb_{J,E}$ is analytic in $E$ for $E\in \hat I_{A,S}$.
\item The reason for assumption {\rm (S3)} is that one only has to consider the limits $\det(\blb_{J,z} S-\one)$ and 
$\det(\blb_{J',z}^* S-\one)$ for $z\to E\in  I_{A,S}$ and these are holomorphic and anti-holomorphic functions for $z$ in the upper half plane.
This is not true for $\det(\blb_{J,z}\blb_{J',z}^*S-\one)$. However, intuitively these functions should not go to zero too often.
\end{enumerate}
\end{remark}

%%%%%%%%%%%%%%%%%%%%%%%%%%%%%%%%%%%%%%%%%%%%%%%%%%%%%%%%%%%%%%%%%%%%%%%%%%%%%%%%%%%%%%%%
%%%%%%%%%%%%%%%%%%%%%%%%%%%%%%%%%%%%%%%%%%%%%%%%%%%%%%%%%%%%%%%%%%%%%%%%%%%%%%%%%%%%%%%%
%%%%%%%%%%%%%%%%%%%%%%%%%%%%%%%%%%%%%%%%%%%%%%%%%%%%%%%%%%%%%%%%%%%%%%%%%%%%%%%%%%%%%%%%

\section{Some Banach spaces} \label{sec-ban} 

We first introduce the important Banach spaces and operators as in \cite{KS, KS2}.
For the supersymmetric background of these definitions see Appendix~\ref{sec-super}.

Let $\I=\{1,\ldots,m\}$ and let $\PP(\I)$ denote the power set of $\I$, i.e. the set of all subsets, 
$\PP(\I)=\{a : a\subset \I\}$. Furthermore,
let $\Sym^+(m)$ denote the set of real, symmetric $m \times m$ matrices $M$ satisfying
$M\geq 0$ in matrix sense, i.e. for all $v\in\RR^m$ one has $v^\top M v\geq 0$.
We define $\Pp$ to be the set of pairs $(\bar a, a)$ of subsets of $\I$ with the same cardinality,
\begin{equation}
\Pp:=\{(\bar a, a)\,: \bar a, a \subset \I,\,|\bar a|=|a| \}\;.
\end{equation}
Moreover, let $n\in\ZZ$,  $n\geq\frac{m}2$, and let
 $\bar\aaa=(\bar a_1,\ldots,\bar a_1),\, \aaa=(a_1,\dots, a_n)\,\in(\PP(\I))^n$ and define
\begin{equation}
\Pp^n:=\{(\bar\aaa,\aaa) \in (\PP(\I))^n \times (\PP(\I))^{n}\,:\,(\bar a_l, a_l)\in \Pp \}\;.
\end{equation}

For functions on $(\Sym^+(m))$ let $\partial_{j,k}$ denote the derivative with respect to the $j,k$- entry of $M$,
$\partial_{j,k} f(M)= \frac{\partial}{\partial M_{j,k}} f(M)$ 
(by symmetry, $\partial_{j,k}=\partial_{k,j}$) and let $\tilde \partial_{j,k}=\frac12(1+\delta_{j,k}) \partial_{j,k}$, where $\delta_{j,k}$ denotes the Kronecker delta symbol, $\delta_{j,k}=0$ for $j\neq k$, and $\delta_{j,j}=1$.
For $(\bar a, a) \in \Pp$ with $a=\{k_1,\ldots,k_c\},\;k_1 < k_2 \ldots < k_c;\;
\bar a =\{\bar k_1,\ldots,\bar k_c\},\;\bar k_1 < \bar k_2 < \ldots \bar k_c$,
we define as in \cite{KS, KS2}
\begin{equation}
 \bpart_{\bar a, a}:=\begin{pmatrix}
       \tilde \partial_{\bar k_1,k_1} & \cdots & \tilde \partial_{\bar k_1,k_c}  \\
        \vdots & \ddots & \vdots \\
	\tilde \partial_{\bar k_c,k_1} & \cdots & \tilde \partial_{\bar k_c, k_c} \end{pmatrix}\;.  
\end{equation}
Now set $D_{\emptyset,\emptyset}$ to be the identity operator and define
\begin{equation}\label{eq-def-Daa}
D_{\bar a,a}:=\det(\bpart_{\bar a, a})\qtx{and}
D_{\bar\aaa,\aaa}:=\prod_{\ell=1}^n D_{\bar a_\ell, a_\ell}\;.
\end{equation}
for $(\bar a, a)\in\Pp$ and $(\bar\aaa, \aaa)\in\Pp^n$.
Something quite important is the following Leibniz-type rule. 
There is a function 
$(\bar\aaa,\aaa,\bar\bbb,\bbb,\bar\bbb',\bbb') \mapsto \sgn(\bar\aaa,\aaa,\bar\bbb,\bbb,\bar\bbb',\bbb')
\in \{-1,0,1\}$ such that
\begin{equation} \label{eq-Leibn0}
D_{\bar\aaa,\aaa}\,(fg) = 
\sum_{ (\bar\bbb,\bbb),(\bar\bbb',\bbb')\in\Pp^n}\;\sgn(\bar\aaa,\aaa,\bar\bbb,\bbb,\bar\bbb',\bbb')
(D_{\bar\bbb,\bbb}\, g)\,(D_{\bar\bbb',\bbb'}\, f)\;.
\end{equation}
The exact expression of the function $\sgn$ above is explained in more detail in Appendix~\ref{sec-super}, cf. \eqref{eq-Leibn}, in many cases one actually has
$\sgn(\bar\aaa,\aaa,\bar\bbb,\bbb,\bar\bbb',\bbb')=0$ (the expression in \eqref{eq-Leibn} does not some over all ($\bar\bbb,\bbb),(\bar\bbb',\bbb') \in \Pp^n$).

For $\bvp,\bvp' \in\RR^{m\times 2n}$, $B\in\CC^{m\times m}$ we define 
\beq
\bvp \cdot B \bvp':=\Tr(\bvp^\top B \bvp')\;,\quad
\bvp^\ot:=\bvp \bvp^\top\,\in\,\Sym^+(m)\,.
\eeq
Let $C_n^\infty(\Sym^+(m))$ denote the set of smooth functions $f$ on the interior of $\Sym^+(m)$ such that
the functions $\bvp \mapsto D_{\bar\aaa, \aaa} f(\bvp^\ot)$ extend to $C^\infty$ functions on $\RR^{m\times 2n}$. (Since $2n\geq m$, $\bvp^\ot$ is in the interior of $\Sym^+(m)$ if $\bvp$ has full rank, hence $\bvp\mapsto f(\bvp^\ot)$ is well defined for a dense open set in $\RR^{m\times 2n}$.
Moreover, let $\Ss_n(\Sym^+(m))$ denote the set of functions where
$\bvp \mapsto D_{\bar\aaa, \aaa} f(\bvp^\ot)$ extends to a Schwartz function.
A smooth function $\bvp\mapsto g(\bvp)$ is a Schwartz function if for any polynomial
function $p(\bvp)$ of the entries of $\bvp$ and any combination of derivatives 
$D^\alpha=\prod_{j,k} \left(\frac{\partial}{\partial \bvp_{j,k}}\right)^{\alpha_{j,k}}$ for a multi-index $\alpha$, one has $\sup_\bvp \left| p(\bvp) D^{\alpha} g(\bvp)\right|<\infty$.

For $f\in \Ss_n(\Sym^+(m))$ we introduce the norms as in \cite{Kl1,KSp,KS,KS2}
\begin{equation} \label{eq-def-norms}
\hn f \hn_p^2\;: =\;
\sum_{(\bar\aaa,\aaa)\in\Pp^n}\; \left\|
2^{|\aaa|} D_{\bar\aaa,\aaa} \;f\,(\bvp^\ot)\right\|^2_{L^p(\RR^{m\times 2n}, d^{2mn}\bvp)}.
\end{equation}
where $|\aaa|=|a_1|+|a_2|+\ldots+|a_n|$.
Now let $\widehat{\Hh}$ be the completion of ${ \Ss_n(\Sym^+(m))}$ with respect to the norm $\hn\cdot\hn_2$, $\widehat{\Hh}$ is a Hilbert space. 
The Banach spaces  $\widehat{\Hh}_p$,  $p \in [1,\infty]$, are defined by
\begin{equation}\label{eq-def-spaces}
\widehat{\Hh}_p : =\;\{f\;\in\;\widehat{\Hh}\;:\; \|f\|_{\widehat{\Hh}_p}\; :=\;\hn f \hn_2 +\hn f \hn_p\;<\; \infty\;\}\;.
\end{equation}
So $f\in\widehat{\Hh}_p$ basically means that for all $(\bar\aaa,\aaa)\in\Pp^n$ the function
$\bvp\mapsto D_{\bar\aaa,\aaa} f(\bvp^\ot)$ is an $L^2$ and $L^p$ function of $\bvp$.

For the same technical reasons as in \cite{KS} we have to work on some specific closed subspaces.
\begin{defini}
\noindent \begin{enumerate}[{\rm (i)}]
\item Let $\Sym_\CC(m)$ denote the complex symmetric, $m\times m$ matrices.
\item For $B\in\Sym_\CC(m)$ with strictly positive real part
(i.e., $\re B > 0$), let $\PE(B)$ denote the vector space spanned by functions $f \in C^\infty_n(\Sym^+(m))$ of the form
$f(M) = p(M) \exp(-\Tr(MB))$, where $p(M)$ is a polynomial in the entries of $M\in\Sym^+(m)$.  
Clearly, $\PE(B) \subset{\Ss_n(\Sym^+(m))}$. 
\item Let $\PE^{(0)}(B)=\{f\in\PE(B):f(0)=0\}$.
\item Define $\PE(m) \subset{ \Ss_n(\Sym^+(m))}$ as the smallest vector spaces containing all vector spaces $\PE(B)$ for all $B\in\Sym_\CC(m)$ with $\re(B)>0$.
\item For $p\in[1,\infty]$ let $\Hh$ and $\Hh_p$ be the closures of $\PE(m)$ in $\widehat\Hh$ and  
$\widehat\Hh_p$, respectively.
\item Let $\Hh^{(0)}=\{f\in\Hh\,:\,f(0)=0\}$ and
for $p\in[1,\infty]$ let $\Hh_p^{(0)}=\{f\in\Hh_p\,:\,f(0)=0\}$.
\end{enumerate}
\end{defini}

Next we show that the last definition actually makes sense, i.e. $f\mapsto f(0)$ can be interpreted as continuous linear functional on $\Hh$ and $\Hh_p$.
As shown in \eqref{eq-int-id0}, integrating over the Grassmann variables in \eqref{eq-sup-id} one obtains 
\begin{equation}\label{eq-int-id}
 \frac{(-1)^{mn}}{\pi^{mn}} \int D_{\III,\III} f(\bvp^\ot) d^{2mn}\bvp=  f(0)\;
\end{equation}
for $f\in \Ss_n(\Sym^+(m))$, where $\III=(\I,\I,\ldots,\I) \in (\PP(\I))^n$ (all entries are the full set $\I$).
Using the Leibniz rule \eqref{eq-Leibn0} and the fact that the product of two $L^2$ functions is $L^1$, one sees that the map
\begin{equation}
 f\mapsto L(f):= \int D_{\III,\III} [e^{-\Tr(\bvp^\ot)} f(\bvp^\ot)] d^{2mn}\bvp
\end{equation}
defines a continuous linear functional on $\Hh$ and any $\Hh_p$, extending the functional $f\mapsto (-1)^{mn}\pi^{mn} f(0)$. 
Therefore, $\Hh^{(0)}$ and $\Hh_p^{(0)}$ are the kernel of $L$ in $\Hh$
and $\Hh_p$, respectively, and we obtain the following.
\begin{lemma}
 $\Hh^{(0)}$ and $\Hh_p^{(0)}$ are closed subspaces of $\Hh$ and $\Hh_p$, respectively, with co-dimension 1.
\end{lemma}

Furthermore, we get the following.

\begin{lemma} \label{lem-PEB} Given any    complex symmetric $m\times m$ matrix $B$  with $\re B > 0$ and $p\in [1,\infty)$,
 $\Hh$ and  $\Hh_p$ are the closures of $\PE(B)$ in $\widehat\Hh$ and  
$\widehat\Hh_p$, respectively.
Consequently, $\Hh^{(0)}$ and $\Hh_p^{(0)}$ are the closures of $\PE^{(0)}(B)$.
\end{lemma}
\begin{proof}
The first statement is precisely Lemma 2.5 in \cite{KS}.
For the second statement, let $f\in\Hh^{(0)}$ or $\Hh_p^{(0)}$, respectively, 
and $f_n\in\PE(B)$ with $f_n\to f$ in $\Hh$ or $\Hh_p$, respectively.
Define $\tilde f_n(M)=f_n(M)-f_n(0) e^{-\Tr(MB)}$, then $\tilde f_n\in\PE^{(0)}(B)$.
Moreover, since $f_n(0)=(-\pi)^{mn}L(f_n)\to 0$, one finds 
$\tilde f_n\to f$ in $\Hh$ or $\Hh_p$, respectively.
\end{proof}

Furthermore, on $\Ss_n(\Sym^+(m))\otimes \Ss_n(\Sym^+(m))$ let us introduce the product norms
\begin{equation}
\hnn g\hnn_p^2 :=
\sum_{(\bar\aaa,\aaa),(\bar\bbb,\bbb) \in \Pp^n} 
\left\| 2^{|\aaa|+|\bbb|} D_{\bar\aaa,\aaa}^{(+)} D_{\bar\bbb,\bbb}^{(-)} g(\bvp_+^\ot,\bvp_-^\ot)
\right\|^2_{L^p(\RR^{4mn}, d^{2mn}\bvp_+ d^{2mn}\bvp_-)}\;,
\end{equation}
%%%
where $D_{\bar\aaa,\aaa}^{(\pm)}$ denotes the operator $D_{\bar\aaa,\aaa}$ with respect to the entry $\bvp^\ot_\pm$.
Then, $\widehat \Kk=\hat\Hh \otimes \hat \Hh$ is the completion of $\Ss_n(\Sym^+(m))\otimes \Ss_n(\Sym^+(m))$ with respect to the norm
$\hnn\cdot\hnn_2$. We further define the Banach spaces  
\begin{equation}
\widehat\Kk_p:= \{ g \in \hat\Kk : \hnn g \hnn_2 + \hnn g \hnn_p < \infty \}
\end{equation}
and $\Kk, \Kk_p$ as the closure of $\PE(m)\otimes\PE(m)$ in $\widehat\Kk$ and $\widehat\Kk_p$, respectively.
Similarly to above, we also define $\Kk^{(0)}$ and $\Kk_p^{(0)}$ as the set of functions
$g(\bvp_+^\ot,\bvp_-^\ot)$ in $\Kk$ and $\Kk_p$, respectively, with $g(0,0)=0$. 
Using the continuous linear map $g\mapsto \int e^{-\Tr(\bvp_+^\ot+\bvp_-^\ot)} g(\bvp_+^\ot,\bvp_-^\ot)d^{2mn}\bvp_+ d^{2mn}\bvp_-$
one can prove that these spaces are closed subspaces of co-dimension 1.

By Lemma~\ref{lem-PEB} $\PE(B)\otimes\PE(C)$ is dense in $\Kk$ and $\Kk_p$, with $p<\infty$, 
for any symmetric $m\times m$  matrices $B,\, C$ with $\re(B)>0$ and $\re(C)>0$. 
Similarly, the vector space sum $\PE^{(0)}(B)\otimes\PE(C) + \PE(B)\otimes\PE^{(0)} (C)$ is dense in $\Kk^{(0)}$ and $\Kk^{(p)}$ for $p<\infty$.

As in \cite{KS,KS2} let us introduce the supersymmetric Fourier transform $T$ acting on $\Ss_n(\Sym^+(m))$ by
\begin{equation} \label{eq-def-T0}
Tf (\bvp'^\ot) = \frac{(-1)^{mn}}{\pi^{mn}} \, \int
e^{\pm i \bvp \cdot \bvp'}\,D_{\III,\III} f(\bvp^\ot)\, d^{2mn}\bvp
\end{equation}
As explained in the calculation \eqref{eq-def-T2}, 
this definition is equivalent to \eqref{eq-def-T} in Appendix~\ref{sec-super} which is the same formula as \cite[eq (2.30)]{KS}.
For this definition it is important that $2n\geq m$, because this insures that the map $\bvp\mapsto \bvp^\ot$ from $\RR^{m\times 2n}$ to
$\Sym^+(m)$ is surjective and hence $Tf\in\Ss_n(\Sym^+(m))$ is well defined.
As $(-\bvp)^\ot=\bvp^\ot$, a change of variables also shows that the right hand side of \eqref{eq-def-T0}
does not depend on the sign of $\pm \bvp\cdot \bvp'$ in the first exponent.

A key identity is the following equation which is derived in \eqref{eq-T-id0}.
Let $B$ be a symmetric matrix and $\im(B)>0$. Then,
\begin{equation}
 \label{eq-T-id}
T(e^{i \Tr(B\bvp^\ot)})=
e^{-\frac{i}4 \Tr(B^{-1} \bvp^\ot)}
\end{equation}

Another important fact is given by Lemma 2.6 in \cite{KS} stating:

\begin{lemma}
\begin{enumerate}[{\rm (i)}]
\item $T$ is unitary on $\widehat \Hh$ and $\Hh$.
\item $T$ is a bounded operator from
$\widehat \Hh_1$ to $\widehat \Hh_\infty$, as well as from $\Hh_1$ to $\Hh_\infty$.
\item By \eqref{eq-int-id} one finds $Tf(0)=f(0)$ for $f\in\Ss_n(\Sym^+(m))$ (see also \eqref{eq-int-id0}). 
Hence $T$ is also unitary
on $\Hh^{(0)}$ and maps $\Hh_1^{(0)}$ to $\Hh_\infty^{(0)}$.
\end{enumerate}
\label{lem-op-T}
\end{lemma}

The operator $\Tt := T \otimes T$ is given by
\begin{align}
& \Tt g(({\bvp'}_+^\ot, {\bvp'}_-^\ot) \\ \notag
& \quad =\frac1{\pi^{2mn}} \int e^{\pm i (\bvp_+'\cdot \bvp_+ \pm \bvp_-'\cdot \bvp_-)} 
D_{\III,\III}^+ D_{\III,\III}^- \,g(\bvp_+^\ot,\bvp_-^\ot)\,d^{2mn} \bvp_+\,d^{2mn} \bvp_-
\end{align}
where $D_{\III,\III}^\pm$ denotes the operator $D_{\III,\III}$ with respect to the entry $\bvp_\pm^\ot$.
$\Tt$ is unitary on $\widehat\Kk$, $\Kk$ and $\Kk^{(0)}$ and it defines a bounded linear map from
$\widehat \Kk_1$ to $\widehat\Kk_\infty$, from $\Kk_1$ to $\Kk_\infty$ and
from $\Kk_1^{(0)}$ to $\Kk_\infty^{(0)}$.

\begin{remark}
The use of the spaces $\Hh^{(0)},\,\Hh_p^{(0)},\,\Kk^{(0)}$ and $\Kk_p^{(0)}$ is new in this work compared to 
\cite{KS}. The restriction to these spaces reduces the spectrum of the Frechet derivative calculated in Section~\ref{sec-frechet}
and helps avoiding an additional assumption on the substitution matrix $S$.
\end{remark}

We will first show some continuous extensions of certain $\Hh$, $\Hh_\infty$, $\Kk$ and $\Kk_\infty$ valued functions depending on the Green's function of the operator $H_\lambda$. 
In order to obtain a fixed point equation in the correct spaces and also in order to get to Theorem~\ref{main2}
we will use certain integral expressions (cf. \eqref{eq-EG-x-D}, \eqref{eq-EGG-x-D}). 
So it will be important that certain functions of the form $\bvp\mapsto D_{\bar\aaa,\aaa} f(\bvp^\ot)$ 
are $L^1$ functions. In order for this function to be in $L^1$, $f\in\Hh_1$ is sufficient. 
As the function $f$ itself will be given by some product, the following observation is important:

By H\"older's inequality, the product of two $L^2$ with some $L^\infty$ functions is an $L^1$ function,
and the product of an $L^2$ with some $L^\infty$ functions is an $L^2$ function.
Therefore, using the Leibniz type rule \eqref{eq-Leibn0} and approximating functions in $\Hh$, $\Hh_p$ by functions in $\PE(m)$, one obtains the following.
\begin{lemma}\label{lemma-prod}
 The product of a function in $\Hh$, or $\Kk$, and finitely many functions in $\Hh_\infty$, or $\Kk_\infty$,
 is in $\Hh_1$, or $\Kk_1$, respectively. Thus, for $r\geq 2$ we have continuous maps
\begin{align}
 (f_1,f_2,\ldots,f_r) \in \Hh \times \Hh_{\infty}^{r-1} &\;\mapsto \;
\prod_{p=1}^r f_p\;\in\Hh_1\;.\\
 (g_1,g_2,\ldots,g_r) \in \Kk \times \Kk_{\infty}^{r-1} &\;\mapsto \;
\prod_{p=1}^r g_p\;\in\Kk_1\;.
\end{align}
\end{lemma}

\begin{remark}\label{rem-Hh_1-lemma}
As one needs the product of two $L^2$ functions to get an $L^1$ function, the assumption $r\geq 2$ is
very important. In view of the fixed point equation developed and analyzed in the next sections,
the assumption {\rm (S1)} will therefore turn out to be crucial. In fact, this is the main reason why
assumption {\rm (S1)} is needed. Together with Lemma~\ref{lem-op-T}~(ii) and Lemma~\ref{lemma-prod} it will assure that the fixed point equations
\eqref{eq-fpz}, \eqref{eq-fpx} are valid in the spaces $\Hh_\infty^s$, $\Kk_\infty^s$
which in turn will be important for using Lemma~\ref{lemma-prod} again together with \eqref{eq-EG-x-D} and
\eqref{eq-EGG-x-D} to obtain Theorem~\ref{main2}.
\end{remark}

%%%%%%%%%%%%%%%%%%%%%%%%%%%%%%%%%%%%%%%%%%%%%%%%%%%%%%%%%%%%%
%%%%%%%%%%%%%%%%%%%%%%%%%%%%%%%%%%%%%%%%%%%%%%%%%%%%%%%%%%%%%

\section{Fixed point equations} \label{sec-fxp}

In this and the following sections let the assumptions (V) and (S1) hold.
For two neighboring sites $x,\, y\in\IS^{(q)}$ let $\IS^{(x|y)}$ denote the rooted tree with root $x$  obtained by
removing the branch from $x$ going through $y$ in $\IS^{(q)}$.
Furthermore, let $ H^{(x|y)}_\lb$ denote the operator $H_\lb$ restricted to $\IS^{(x|y)}$ with Dirichlet boundary
conditions and similar to \eqref{eq-def-G^x} let
\begin{equation}
 G_\lb^{(x|y)}(z):=
\left[ \langle x,j | (H_\lb^{(x|y)}-z)^{-1} | x,k\rangle \right]_{j,k}\;\in\,\CC^{m\times m}\;.
\end{equation}
For simplicity, we will denote the Green's functions at the roots by
\begin{equation}
 G_\lb^{(q)}(z):=
 G_\lb^{[0_q]}(z)=
\left[ \langle 0_q,j | (H_\lb-z)^{-1} | 0_q,k\rangle \right]_{j,k}\;\in\,\CC^{m\times m}\;.
\end{equation}

In \cite{KS} we used the operator $\frac 12\Delta$ instead of $\Delta$ on regular trees, so in order to
easily refer to these formulas, let us define the following $m\times m$ matrix valued Green's functions by
\begin{align} \label{tilde-G}
\tilde G^{[x]}_\lb(z)&:= 2 G_\lb^{[x]}(2z) = \left[ \langle x,j| (\tfrac12 H_\lb - z)^{-1} |x, k \rangle \right]_{j,k\in\I}\;,\;\;
\tilde G^{(q)}_\lb(z):=\tilde G^{[0_q]}_\lb(z)\;\\\label{tilde-G1}
 \tilde G^{(x|y)}_{\lb}(z) &:= 2 G^{(x|y)}_\lb(2z) = \left[ \langle 0_q,j| (\tfrac12 H^{(x|y)}_\lb - z)^{-1} |0_q, k \rangle \right]_{j,k\in\I}\;.
\end{align}
As
\beq
\tfrac12 H_\lb=\tfrac12 \Delta \otimes \one + \one\otimes \tfrac12 A + \bigoplus_{x\in\IS} \tfrac12 \lb  V(x)
\eeq
this means one needs to replace $A$ and $V(x)$ by $\frac12 A$ and $\frac12 V(x)$ compared to the formulas in \cite{KS, KS2}.
Therefore, for any $z=E+i\eta$ in the upper half plane, i.e., $\eta>0$, and label $q=1,\ldots, s$, define the functions  
$\ze^{(q)}_{\lb,z}$, $\ze^{(x|y)}_{\lb,z}$  on $\Sym^+(m)$ and 
$\xi^ {(q)}_{\lb,z}$, $\xi^ {(x|y)}_{\lb,z}$ on $\Sym^+(m)\times \Sym^+(m)$ by
\begin{align}\label{eq-zeta}
\ze^{(q)}_{\lb,z} (\bvp^{\od 2}) &:=  \E\left( e^{ \frac{i}{4} \Tr( \tilde G^{(q)}_{\lb}(\tfrac{z}{2})\,\bvp^{\od 2})}\right)
 = \E\left(\e^{\frac{i}2 \Tr(G_{\lb}^{(q)}(z)\,\bvp^\ot)} \right)\;, \\
\label{eq-zeta-g}
\ze^{(x|y)}_{\lb,z} (\bvp^{\od 2}) &:=  \E\left( e^{ \frac{i}{4} \Tr( \tilde G^{(x|y)}_{\lb}(\tfrac{z}{2})\,\bvp^{\od 2})}\right)
 = \E\left(\e^{\frac{i}2 \Tr(G_{\lb}^{(q)}(z)\,\bvp^\ot)} \right)\;, \\
 \xi^{(q)}_{\lb,z} ( \bvp_+^\ot\,,\, \bvp_-^\ot) &:= 
\E \left(e^{ \frac{i}2 \Tr\left ( G_{\lb}^{(q)}(z) \,\bvp_+^\ot\, - \,
\overline{ G_{\lb}^{(q)}(z) }\,\bvp_-^\ot\right)}\right)    \label{xia} \;,\\
 \xi^{(x|y)}_{\lb,z} ( \bvp_+^\ot\,,\, \bvp_-^\ot) &:= 
\E \left(e^{ \frac{i}2 \Tr\left ( G_{\lb}^{(x|y)}(z) \,\bvp_+^\ot\, - \,
\overline{ G_{\lb}^{(x|y)}(z) }\,\bvp_-^\ot\right)}\right)    \label{xia-g} \;.
\end{align}
%%%
Moreover, similar to \cite{KS} let us introduce
the operators 
\begin{align}\label{eq-def-B}
B_{\lb,z} &= \Mm\left(e^{i\Tr(\frac{1}2(z-A) \bvp^\ot) }\; h(\tfrac12 \lb  \bvp ^{\od 2})\right)\,,\\
\Bb_{\lb,z} &=\Mm\left( e^{i\Tr(\frac{1}2(z-A) \bvp_+^\ot-(\frac12(\bar z-A)\bvp_-^\ot)) }\; h(\tfrac12 \lb  (\bvp_+^{\ot}-\bvp_-^\ot))
\;\right)
\end{align}
where the function $h$ is the Fourier transform of the distribution of $V(x)$ as given in \eqref{eq-def-h}.
Here, for a given function $g$ defined on $\Sym^+(m)$ and $\tilde g$ defined on
$\Sym^+(m)\times \Sym^+(m)$ we use $\Mm(g)$, $\Mm(g(\bvp^\ot))$ and $\Mm(\tilde g), \Mm(\tilde g(\bvp_+^\ot, \bvp_-^\ot)$, respectively, 
to denote the multiplication-operator given multiplying by $g(\bvp^\ot)$ and $\tilde g(\bvp_+^\ot,\bvp_-^\ot)$, respectively. This means,  
\begin{equation}
(\Mm(g)f)(\bvp^\ot)\; = \;g(\bvp^\ot)f(\bvp^\ot)\;.
\end{equation}

Let $N_p$ denote the set of neighbors of the root $0_p$, i.e. $N_p=\{x: d(x,0_p)=1\}$.
Analogous to \cite[eq. (3.13)]{KS}, using the supersymmetric replica trick (cf. Appendix~\ref{app-gr}) one obtains
\begin{equation}\label{eq-recursion}
 e^{\frac{i}{4}\,\Tr\left(\tilde G^{(p)}_{\lambda}(z)\bvp^\ot\right)} = 
T \Mm(e^{i  \Tr((z-\lambda \frac12 V(x)-\frac 12 A)\bvp^\ot)})\left(
\prod\limits_{x\in N_p} e^{\frac{i}4 \Tr(\tilde G^{(x|0_p)}_{\lambda}(z) \bvp^\ot)} \right)
\end{equation}
In view of \eqref{eq-T-id} this equation is equivalent to
\begin{equation}\label{eq-recursion-G}
 G^{(p)}_{\lambda}(z)=
-\left(z-\lambda V(0_p)-A+\sum_{x\in N_p} G^{(x|0_p)}_\lambda(z) \right)^{-1}
\end{equation}
which is a well known recursion relation for the Green's functions that can be obtained from the resolvent identity.
It is used in many articles and gives an alternative proof of \eqref{eq-recursion}.
Taking expectations in \eqref{eq-recursion} and replacing $z$ by $\frac z2$ this gives
\begin{equation} \label{eq-zeta-recursion}
\ze^{(p)}_{\lb,z} = T B_{\lb,z} \left( \prod_{x\in N_p}\ze^{(x|0_p)}_{\lb,z}\right) = T B_{\lb,z} \left(  \prod_{q=1}^s
[\ze^{(q)}_{\lb,z}]^{S_{p,q}} \right)\;
\end{equation}
which is the analogue of \cite[eq. (3.16) and eq. (3.23)]{KS}.
In the equation above we used that $\ze^{(x|0_p)}=\ze^{(q)}$ if $x\in N_p$ and the label of $x$ is $q$.
This follows from the fact that the potential $(V(x))_{x\in\TT}$ is independent identically distributed and that for $x\in N_p$ with label $q$
the tree $\IS^{(x|0_p)}$ is equivalent to $\IS^{(q)}$.

Analogously, like in \cite[eq.~(4.4) and (4.12)]{KS}, one obtains
\begin{equation}
\xi^{(p)}_{\lb,z} = 
\Tt \Bb_{\lb,z} \left( \prod_{x\in N_p}\xi^{(x|0_p)}_{\lb,z}\right) =
\Tt \Bb_{\lb,z} \left(  \prod_{q=1}^s
[\xi^{(q)}_{\lb,z}]^{S_{p,q}} \right)\;.\label{eq-xi-recursion}
\end{equation}

 Recall that we assume without loss of generality that $A$ is diagonal.
Then, the Hamiltonian $H_0$ (i.e., $\lambda=0$) splits into a direct sum of shifted Laplacians on $m$ copies of the
forest $\IS$. The Laplacians are   shifted by the energies $a_i$, $i=1,\ldots,m$, where $A=\diag(a_1,\ldots,a_m)$.
Therefore, in the free case, $\lambda=0$, one obtains
%%%%%%%
\begin{equation}
\ze^{(q)}_{0,z} (\bvp^{\od 2}) =  
%\prod_{k=1}^m  e^{\frac{i}2\Gamma^{(q)}_{z-a_k} \varphi_k^2} =
e^{\frac{i}2 \Tr(A^{(q)}_z \bvp^\ot)}\;,\;\;\text{where}
\;\;
A^{(q)}_z := \diag(\Gamma^{(q)}_{z-a_1},\Gamma^{(q)}_{z-a_2},\ldots,\Gamma^{(q)}_{z-a_m})\;.
\label{eq-ze0}
\end{equation}
By \eqref{eq-def-Iq} and \eqref{eq-def-hat-I}, for  $E\in \hat I_{A,S}\subset I_{A,S}$ the point-wise limits 
\begin{align}
\ze^{(q)}_{0,E} (\bvp^{\od 2}) &{:=} \lim_{\eta \downarrow 0} \ze^{(q)}_{0,E+i\eta} (\bvp^{\od 2})
=e^{\frac{i}2 \Tr(A^{(q)}_E\,\bvp^\ot)}, \label{eq-ze00} \\
\xi^{(q)}_{0,E} (\bvp_+^{\ot},\bvp_-^\ot) &{:=} \lim_{\eta \downarrow 0} \xi^{(q)}_{0,E+i\eta} (\bvp_+^{\ot},\bvp_-^\ot) =
e^{\frac{i}2 \Tr(A^{(q)}_E\,\bvp_+^\ot - \bar A^{(q)}_E\,\bvp_-^\ot)}
\end{align}
exist for any $q=1,\ldots,s$, where
\begin{equation}
 A_E^{(q)}=\lim_{\eta\downarrow 0} A_{E+i\eta}^{(q)}\;,\qquad
[A_E^{(q)}]_{jk}=\delta_{jk}\,\Gamma^{(q)}_{E-a_j} \label{eq-def-AE}
\end{equation}
are diagonal $m\times m$ matrices with strictly positive imaginary part.
Note that this combined with \eqref{eq-def-blb} leads to
\begin{equation} \label{eq-rel-blb-A}
\theta_{J,E}^{(q)}=\prod_{\substack{j,k\in\{1,\ldots,m\}\\j\leq k}} \left[[A_E^{(q)}]_{j,j} [A_E^{(q)}]_{k,k} \right]^{J_{jk}}\;.
\end{equation}

In order to write \eqref{eq-zeta-recursion} and \eqref{eq-xi-recursion} in a more compact way, let us introduce
the column-vectors
\begin{equation}
\vze_{\lb,z}=(\ze_{\lb,z}^{(1)},\ze_{\lb,z}^{(2)}, \ldots,\ze_{\lb,z}^{(s)})^\top\,,\quad
\vxi_{\lb,z}=(\xi_{\lb,z}^{(1)},\xi_{\lb,z}^{(2)}, \ldots,\xi_{\lb,z}^{(s)})^\top\,,\;
\end{equation}
and for a vector $\vec{v}=(v_1,\ldots,v_s)^\top \in\CC^s$ and an $s\times s$ matrix $M$  we define the 
notation $\vec{v}^M$ by
\begin{equation}
 \vec{v}^M\,=\,\vec{w}\in\CC^s\quad:\Leftrightarrow\quad
w_p=\prod_{q=1}^s (v_q)^{M_{pq}}\;.
\end{equation}

Then  \eqref{eq-zeta-recursion} and \eqref{eq-xi-recursion} can be written as 
$\vze_{\lb,z} = T B_{\lambda,z} (\vze_{\lb,z})^S$ 
and $\vxi_{\lb,z} = \Tt \Bb_{\lambda,z} (\vxi_{\lb,z})^S$.
Here, $B_{\lambda,z},\,T,\,\Bb_{\lambda,z},\,\Tt$  act on each component of the vector.
The crucial observation is that these are fixed point equations in $\Hh_\infty^s$ and $\Kk_\infty^s$.
Analogous to \cite[Proposition 3.2]{KS} one obtains

\begin{prop} \label{zeta} We have:
\begin{enumerate}
\item[{\rm (i)}] For $\eta=\im z\geq 0$ the operator $B_{\lb,z}$ is a bounded operator on $\widehat\Hh_1$, leaving  $\Hh_1$ invariant. 
The map 
\begin{equation}
(\lambda,E,\eta,\vec{f}) \mapsto\, T B_{\lambda,E+i\eta} \vec{f}\,^S
\end{equation}
is a continuous map from 
$\RR\times\RR\times[0,\infty)\times \Hh^s_\infty$ to $\Hh^s_\infty$.

\item[{\rm (ii)}]   $\vze_{\lb,z} \in {\Hh}^s_\infty$ for all $\lb \in \RR$ and $z=E+i\eta$ with 
$\eta >0$.   The map $(\lb,E, \eta) \to \vze_{\lb,E +i\eta}$ is continuous from
$\RR \times \RR \times (0,\infty)$ to $ {\Hh}^s_\infty$.

\item[{\rm (iii)}]  If $E\,\in\,\hat I_{A,S}$, then  $\vze_{0,E} \in (\PE(m))^s \subset {\Hh}^s_\infty$ and
\begin{equation}
\lim_{\eta \downarrow 0} \vze_{0,E+i\eta}\; = \;\vze_{0,E} \;\;\;\mbox{in}\;\; {\Hh}^s_\infty  \;.   
\end{equation}

\item[{\rm (iv)}] {The equality}  \eqref{eq-zeta-recursion} can be rewritten as a fixed point equation in 
${\Hh}^s_\infty$: 
\begin{equation} \label{eq-fpz}
\vze_{\lb,z}  \; = \;TB_{\lb,z} (\vze_{\lb,z})^S %\,=\,F(\lb,E,\eta,\vze_{\lb,z})\;,
\end{equation}
valid for all $\lb \in \RR$ and $z=E+i\eta$ with $\eta >0$, and also valid for  $\lb=0$ and
 $z=E$ with $E\,\in\, \hat I_{A,S}$. 
\end{enumerate}
\end{prop}

\begin{proof}
(i) The fact that $B_{\lambda,z}$ is a bounded operator on $\widehat\Hh_1$ and $\Hh_1$ as well as the fact that $(\lambda,E,\eta,f)\mapsto B_{\lb,z} f$ is a continuous map from
$\RR\times\RR\times[0,\infty)\times \Hh_1$ to $\Hh_1$ are already proved in \cite[Proposition~3.2]{KS}.
As $B_{\lambda,z}$ is a multiplication operator, including the multiplication by the Fourier transform $h$ of the distribution of the potential $V(x)$ (cf. definitions \eqref{eq-def-h} and \eqref{eq-def-B}), assumption {\rm (V)} is important for this observation.

By assumption {\rm (S1)} each component in $\vec{f}^S$ is the product of at least two factors,
hence by Lemma~\ref{lemma-prod}, $\vec{f}\mapsto\vec{f}^S$ defines a continuous map 
from $\Hh_\infty^s$ to $\Hh_1^s$. 
These two facts together with Lemma~\ref{lem-op-T}~(ii) immediately imply (i).

To get (ii) note that 
for fixed potential and $\eta>0$ the dependence of $G^{(q)}_{\lambda}(E+i\eta)$ on $\lambda, E,\eta$ is continuous and $\|G_{\lambda}(E+i\eta)\|\leq\frac{1}{\eta}$.
Also, as long as $\eta>0$, the multiplication operator $B_{\lambda,E+i\eta}$ multiplies by
a exponential decaying function (in $\bvp$, cf. \eqref{eq-def-B}), and the exponential decay is uniform in
a neighborhood of $z$. Therefore, $B_{\lambda,E+i\eta} (\vze_{\lambda,E+i\eta})^S \in \Hh_1^s$ and by Dominated Convergence, the dependence on $(\lambda,E,\eta)$ is continuous.
Using \eqref{eq-zeta-recursion} (which can be written as \eqref{eq-fpz}) and Lemma~\ref{lem-op-T}~(ii),
part (ii) now follows.

For part (iii) and (iv) the important fact is that by definition, 
$E \in \hat I_{A,S} \subset I_{A,S}$ assures that for $\lambda=0$ (no random potential) the limit
$G^{(q)}_0(E):=\lim_{\eta\downarrow 0} G^{(q)}_{0}(E+i\eta))$ exists for all $q$ and has a positive definite imaginary part. Hence, the limit $\vze_{0,E}(\bvp^\ot)=e^{\frac{i}2 \Tr(G^{(q)}_0(E)\bvp^\ot)}\in\PE(m)$ exists pointwise (in $\bvp^\ot$) and all derivatives
$D_{\bar\aaa,\aaa} \vze_{0,E}(\bvp^\ot)$
are exponentially decaying functions of $\bvp$, the decay is uniform in $z$ in a neighborhood of $E+i0$ in the upper half plane. 
Dominated Convergence gives $\vze_{0,E+i\eta} \to \vze_{0,E}$ in $\Hh^s$ and
$B_{0,E+i\eta} (\vze_{0,E+i\eta})^S \to B_{0,E} (\vze_{0,E})^S$ in $\Hh_1^s$.
By taking limits in \eqref{eq-zeta-recursion} we first obtain $\vze_{0,E}=TB_{0,E}\vze^S_{0,E}$ and using Lemma~\ref{lem-op-T}, part (iii) and (iv) follow.  
\end{proof}

Similarly, as in \cite[Proposition 4.2]{KS} one also obtains the analogue results for the function $\vxi_{\lambda,z}$.
%@xita
\begin{prop}  \label{xita} We have:
\begin{enumerate}
\item[{\rm (i)}] For $\eta=\im z\geq 0$ the operator $\cB_{\lb,z}$ is a bounded operator on $\Kk_1$.
Furthermore, the map
\begin{equation}
(\lambda,E,\eta,\vec{g})\mapsto \Tt\cB_{\lambda,E+i\eta} \vec{g}^S
\end{equation} 
is a continuous map
from $\RR\times\RR\times [0,\infty)\times \Kk^s_\infty$ to $\Kk^s_\infty$.

\item[{\rm (ii)}] $\vxi_{\lb,z} \in {\Kk}^s_\infty$ for all $\lb \in \RR$ and $z=E+i\eta$ with 
$\eta >0$.   The map $(\lb,E, \eta) \to \vxi_{\lb,E +i\eta}$ is continuous from
$\RR \times \RR \times (0,\infty)$ to $ {\Kk}^s_\infty$.

\item[{\rm (iii)}]  If $E\,\in\,\hat I_{A,S}$, then
 $\vxi_{0,E} \in {\Kk}^s_\infty$ and
\begin{equation}
\lim_{\eta \downarrow 0} \vxi_{0,E+i\eta}\; = \;\vxi_{0,E} \;\;\;\mbox{in}\;\; {\Kk}^s_\infty   \;. 
\end{equation}

\item[{\rm (iv)}]  {The equality} \eqref{eq-xi-recursion} can be rewritten as a fixed point equation in 
${\Kk}^s_\infty$: 
\begin{equation}\label{eq-fpx}
\vxi_{\lb,z}  \; = \;{\Tt}{\cB}_{\lb,z} (\vxi_{\lb,z})^S \;,
\end{equation}
valid for all $\lb \in \RR$ and $z=E+i\eta$ with $\eta >0$, and also valid for  $\lb=0$ and
 $z=E$ with  $E\in \hat I_{A,S}$. 
 \end{enumerate}
\end{prop}

%%%%%%%%%%%%%%%%%%%%%%%%%%%%%%%%%%%%%%%%%%%%%%%%%%%%%%%%%%%%%%%%%%%%%%%%%%%%%%%%%%%%%%%%
%%%%%%%%%%%%%%%%%%%%%%%%%%%%%%%%%%%%%%%%%%%%%%%%%%%%%%%%%%%%%%%%%%%%%%%%%%%%%%%%%%%%%%%%
%%%%%%%%%%%%%%%%%%%%%%%%%%%%%%%%%%%%%%%%%%%%%%%%%%%%%%%%%%%%%%%%%%%%%%%%%%%%%%%%%%%%%%%%
%%%%%%%%%%%%%%%%%%%%%%%%%%%%%%%%%%%%%%%%%%%%%%%%%%%%%%%%%%%%%%%%%%%%%%%%%%%%%%%%%%%%%%%%

\section{Frechet derivative and its spectrum} \label{sec-frechet}

%%%%%%%%%%%%%%%%%%%%%%%%%%%%%%%%%%%%%%%%%%%%%%%%%%%%%%%%%%%%%%%%%%%%%%%%%%%%%%%%%%%%%%%%
%%%%%%%%%%%%%%%%%%%%%%%%%%%%%%%%%%%%%%%%%%%%%%%%%%%%%%%%%%%%%%%%%%%%%%%%%%%%%%%%%%%%%%%%
%%%%%%%%%%%%%%%%%%%%%%%%%%%%%%%%%%%%%%%%%%%%%%%%%%%%%%%%%%%%%%%%%%%%%%%%%%%%%%%%%%%%%%%%
%%%%%%%%%%%%%%%%%%%%%%%%%%%%%%%%%%%%%%%%%%%%%%%%%%%%%%%%%%%%%%%%%%%%%%%%%%%%%%%%%%%%%%%%

In this section we will analyze the fixed point equations \eqref{eq-fpz} and \eqref{eq-fpx} in more detail.
Recall that $\Delta(m,\ZZ_+)$ denotes the collection of  $m\times m$  upper triangular matrices with 
non-negative integer entries and for $J=(J_{j,k})_{j,k}\in \Delta(m,\ZZ_+)$ we defined $|J| = \sum_{j\leq k} J_{j,k}$.
Let 
\begin{equation}
\epsilon(\bvp^\ot)=e^{-\Tr(\bvp^\ot)}\;,\quad
\vec\epsilon\,(\bvp^\ot):=\left(\epsilon(\bvp^\ot),\epsilon(\bvp^\ot),\ldots,\epsilon(\bvp^\ot)\right)^\top\,\in\,[\PE(m)]^s\;
\end{equation}
and define the map $F: \RR\times\RR\times [0,\infty)\times \big(\Hh^{(0)}_\infty \big)^s \to \big(\Hh^{(0)}_\infty \big)^s$ by
\begin{equation}
 F(\lb,E,\eta,\vec f):=TB_{\lb,z} \left((\vec\epsilon +\vec f\,)^S\right)- \vec\epsilon\;.
\end{equation}
By Proposition~\ref{zeta} this is a continuous map and using Lemma~\ref{lem-op-T}~(iii) one obtains that indeed
$F(\lb,E,\eta,\vec f)\in\big(\Hh^{(0)}_\infty \big)^s$
for $\vec f\in \big(\Hh^{(0)}_\infty \big)^s$. Moreover, one finds 
\begin{equation}\label{eq-fxp}
\vze_{\lb,z} -\vec\epsilon \in \big(\Hh^{(0)}_\infty \big)^s \qtx{and}
F(\lb,E,\eta,\vze_{\lb,z} -\vec\epsilon)\;=\;\vze_{\lb,z}-\vec\epsilon
\end{equation}
where the second equation follows from \eqref{eq-fpz}.
The following Lemma corresponds to \cite[Lemma 5.1]{KS}.

\begin{lemma}\label{lem-DF}
\noindent {\rm (i)} The map $F$ is continuous and Frechet-differentiable w.r.t. $\vec{f}\in\Hh_\infty^s$. The derivative $F_{\vec{f}}$ for $\vec{f}\in\big(\Hh^{(0)}_\infty \big)^s$
is a bounded operator on 
$\big(\Hh^{(0)}_\infty \big)^s$ and extends naturally to a bounded operator on $\big(\Hh^{(0)} \big)^s$ which we will also denote as $F_{\vec{f}}$\\
\noindent {\rm (ii)} For $E\in \hat I_{A,S}$ let $C_E=F_{\vec{f}}(0,E,0,\vec{\ze}_{0,E}-\vec\epsilon)$, then $C_E^2$ is a compact operator on $\big(\Hh^{(0)}\big)^s$ and $\big(\Hh^{(0)}_\infty\big)^s$.\\
\noindent {\rm (iii)}
The spectrum of $C_E$ as an operator on the Hilbert space  $\big(\Hh^{(0)}\big)^s$ is given by the eigenvalues 
of the matrices $\blb_{J,E} S$ for $|J|\geq 1$, and the accumulation point $0$.
This means, denoting the spectrum of $C_E$ on $\big(\Hh^{(0)}\big)^s$ by $\sigma_{\Hh^{(0)}}$ one obtains
\begin{equation} \label{eq-lbj}
\sigma_{\Hh^{(0)}}(C_E)=\bigcup_{\substack{J\in\Delta(m,\ZZ_+)\\ |J|\geq 1}} \sigma(\blb_{J,E} S)\;\cup\;\{0\}\;
\end{equation}
where 
$\blb_{J,E}$ are the matrices as defined in \eqref{eq-def-blb}.  
In particular, by the definition of $\hat I_{A,S}$ one has for $E\in \hat I_{A,S}$
\beq \label{eq-spectrum-CE}
1 \notin \sigma_{\Hh^{(0)}}(C_E)
\eeq

\noindent {\rm (iv)} The spectrum of  $C_E$ as an operator on 
$\big(\Hh^{(0)}_\infty\big)^s$, denoted by $\sigma_{\Hh^{(0)}_\infty}(C_E)$,
is the same as  its spectrum as an operator on $\big(\Hh^{(0)}\big)^s$:  
\beq \label{eq-spectrum-CEinfty}
\sigma_{\Hh^{(0)}_\infty}(C_E) = \sigma_{\Hh^{(0)}}(C_E).
\eeq
\end{lemma}

\noindent{\bf Proof.} 
(i) The derivative $F_{\vec{f}}$ can be written
as a matrix of operators. 
Considering the $p$-th entry of $F$, we get formally
\begin{equation}\label{eq-der-F0}
[F_{\vec{f}}]_{p,q}=\partial_{f_q} F_p = 
TB_{\lambda,z} \Mm\left(S_{p,q} \prod_{r=1}^s ((\epsilon+f_r)^{S_{p,r}-\delta_{q,r}})\right)
\end{equation}
Let us define
\begin{equation}\label{eq-def-diagf}
\diag(\vec{f}\,)\,=\,\diag(f_1,\ldots,f_s)\,=\,\left(\begin{smallmatrix}
                                                f_1 \\ & \ddots \\ & & f_s
                                               \end{smallmatrix}\right)
\end{equation}
which will be considered as an operator acting by matrix multiplication on a vector of functions.
Then, \eqref{eq-der-F0} can be written as
\begin{equation}\label{eq-der-F1}
 F_{\vec{f}}\,=\,TB_{\lambda,z} \Mm(\diag((\vec\epsilon+\vec{f}\,)^S)\,S\,\diag((\vec\epsilon+\vec{f}\,)^{-\one}))\,,
\end{equation}
where $\Mm(D)$ denotes the multiplication operator $\vec{g}\mapsto D\vec{g}$ for a matrix  valued function $D$.

Even despite the term $(\vec\epsilon+\vec f)^{-\one}$,
one does not divide by any of the components of $\vec\epsilon+\vec f$. The terms in the denominators always cancel.
This can be seen in \eqref{eq-der-F0}. 
Because $S_{p,r}-\delta_{q,r}=-1$ is equivalent to $r=q$ and
$S_{p,q}=0$, one finds $[F_{\vec f}]_{p,q}=0$ in this case. 
If $S_{p,q}\neq 0$, then $S_{p,q}-\delta_{q,r}\geq 0$ for all $r$.
By assumption (S1) the product on the right hand side of \eqref{eq-der-F0} 
has at least one factor in this case. Thus for $\vec f\in\big(\Hh^{(0)}_\infty\big)^s$, $[F_{\vec{f}}]_{p,q} $ is 
the composition of $T$ and a multiplication operator by an $\Hh_\infty$ function.
Therefore, by Lemma~\ref{lem-op-T} and Lemma~\ref{lemma-prod}, $[F_{\vec f}]_{p,q}$ defines
a bounded linear operator on $\Hh^{(0)}$ and on $\Hh^{(0)}_\infty$
for $\vec f \in \big(\Hh^{(0)}_\infty\big)^s$.
This implies the Frechet differentiability and claim (i) follows.

To get (ii) note that $C_E\vec f= TB_{0,E} \diag((\vze_{0,E})^S) S\, \diag((\vze_{0,E})^{-\one})\,\vec f$ and that
$C_E$ and $C_E^2$ can be seen as $s\times s$ matrices of operators.
Compactness of $C_E^2$ then follows from compactness of the matrix entries.
This can be proved completely analogous to Lemma~5.1~(i) in \cite{KS}. There one shows
that for functions $f_1,f_2 \in \PE(m)$ the operator $\Mm(f_1) T \Mm(f_1)$
is compact on $\Hh$ and $\Hh_\infty$. As $\vze_{0,E} \in (\PE(m))^s$ for $E\in \hat I_{A,S}$,
the entries of $C_E^2$ are sums of operators of the form
$T\Mm(f_1)T\Mm(f_2)$ with $f_1, f_2\in\PE(m)$.

To obtain (iii), let $g\in C_n^\infty(\Sym^+(m))$ and let us start with the identity   
\begin{align}
 C_E\;  (g \,\ze^{(q)}_{0,E} \vec{e}_q )&=
TB_{0,E} \Big(\diag((\vze_{0,E})^S) \,S \,\diag((\vze_{0,E})^{-\one}) g \;\diag(\vze_{0,E}) \vec{e}_q \Big) \notag \\
&=TB_{0,E}\, \Big( g \,\diag((\vze_{0,E})^S) \,S\, \vec{e}_q \Big) \label{eq-start-id}
\end{align}
where $\vec{e}_q \in \CC^s$ denotes the $q$-th canonical basis vector, $(\vec{e}_q)_p = \delta_{p,q}$.
Using $g(\bvp^\ot)=\exp({\frac12 t \Tr(M\bvp^\ot)})$,  \eqref{eq-start-id} and 
\eqref{eq-T-id} imply for the $p$-th component
\begin{align}
&\left[C_E\; e^{\frac i2 t \Tr(M\bvp^\ot)} \zeta^{(q)}_{0,E}(\bvp^\ot) \vec{e}_q\right]_p \notag 
= S_{p,q} T\left( e^{\frac i2 \Tr\big((E-A+\sum_{r=1}^s S_{p,r} A^{(r)}_E-tM)\bvp^\ot\big)}\right) \\ 
&\qquad\qquad\qquad\qquad\qquad\qquad = S_{p,q} \,e^{\frac i2  \Tr\big(\big(A-E-\sum_{r=1}^s S_{p,r} A^{(r)}_E-tM\big)^{-1} \bvp^\ot\big)} \label{eq-CE}\;.
\end{align}

Let $\Pp_u(\bvp^\ot)$ denote the set of homogeneous
polynomials of degree $u$ in the entries of $\bvp^\ot$, together with the zero polynomial to make it a vector-space.
{Furthermore, let $\Pp_{\leq u}(\bvp^\ot)$ and $\Pp_{<u}(\bvp^\ot)$ denote the polynomials in the entries of $\bvp^\ot$ of degree
smaller or equal to $u$ and strictly less than $u$, respectively.}

Using \eqref{eq-ze00} and \eqref{eq-T-id}, a Taylor expansion with respect to $t$ of the right hand side of \eqref{eq-CE} gives
\begin{align}\notag
& e^{\frac{i}{2} \Tr\left((A-E-\sum_{r=1}^s S_{p,r} A^{(r)}_E-tM)^{-1} \bvp^\ot\right)} = 
e^{\frac{i}{2} \Tr(A^{(p)}_E\bvp^\ot)}
e^{\frac{i}{2} \sum_{u=1}^\infty t^u \Tr\left((A^{(p)}_E M)^u (A^{(p)}_E) \bvp^\ot\right)} \\
 & \qquad =
\zeta^{(p)}_{0,E}(\bvp^\ot)\left[1+\sum_{u=1}^\infty \frac{(\frac{i}2 t)^u}{u!}\left(\left[\Tr(A^{(p)}_E M A^{(p)}_E \bvp^\ot)\right]^u 
+p_{u,M}(\bvp^\ot) \right) \right],
\end{align}
where $p_{u,M}\in\Pp_{<u}(\bvp^\ot)$.
Performing a Taylor expansion of the left hand side of
\eqref{eq-CE} and comparing terms leads to
\begin{align}\label{eq-CE-2}
 &\left[C_E\left(\Tr(M \bvp^\ot)\right)^u  \diag(\vze_{0,E}(\bvp^\ot)) \vec{e}_q \right]_p \\ \notag
&\qquad=S_{p,q} \left[\left(\Tr(A_E M A_E \bvp^\ot)\right)^u \,+\,p_{u,M}(\bvp^\ot) \right]\zeta^{(p)}_{0,E}(\bvp^\ot)\;.
\end{align}
Since the natural projection from $\Pp_{\leq u}(\bvp^\ot)$ onto $\Pp_{u}(\bvp^\ot)$ as well as the operator $C_E$ are linear,
the map $[\Tr(M\bvp^\ot)]^u \mapsto [\Tr(A_E M A_E \bvp^\ot)]^u$, varying $M$, 
can be extended to a linear map on $\Pp_u(\bvp^\ot)$.
Using all real symmetric matrices $M$, the polynomials of the form $[\Tr(M\bvp^\ot)]^u$ 
span $\Pp_u(\bvp^\ot)$. Hence the extension is unique.
To expand these homogeneous polynomials let us define for $J\in\Delta(m,\ZZ_+)$
\beq
P_J(\bvp^\ot):=\prod_{\substack{j,k\in\{1,\ldots,m\}\\j\leq k}}
\left[(\bvp^\ot)_{jk}\right]^{J_{jk}}\,.
\eeq
Then one has
\beq \label{eq-expand1}
[\Tr(M\bvp^\ot)]^u=\sum_{\substack{j_1,\ldots,j_s\\k_1,\ldots,k_s}} \prod_{i=1}^u M_{j_i,k_i} 
(\bvp^\ot)_{j_i,k_i}
 = \sum_{\substack{J\in\Delta(m,\ZZ_+)\\|J|=u}} \!\!c(M,J) P_J(\bvp^\ot)\,,
\eeq
where the latter equation defines the coefficients $c(M,J)$.
Similarly, using that $A^{(p)}_E$ is diagonal as well as \eqref{eq-rel-blb-A} one obtains
\begin{align} \label{eq-expand2}
[\Tr( A^{(p)}_E M A^{(p)}_E \bvp^\ot)]^u &=\sum_{\substack{j_1,\ldots,j_u\\k_1,\ldots,k_u}} \prod_{i=1}^u 
M_{j_i,k_i}\,(A^{(p)}_E)_{j_i,j_i} (A^{(p)}_E)_{k_i,k_i} (\bvp^\ot)_{j_i,k_i} \nonumber \\
&=\sum_{\substack{J\in\Delta(m,\ZZ_+)\\|J|=u}} \theta^{(p)}_{J,E} \, c(M,J) P_J(\bvp^\ot)\,.
\end{align}
Thus, we conclude that
\begin{equation}
 \left[C_E P_J(\bvp^\ot)\,\zeta^{(q)}_{0,E}(\bvp^\ot)\vec{e}_q\right]_p  =  
\zeta^{(p)}_{0,E}(\bvp^\ot)\left[P_J(\bvp^\ot)\theta^{(p)}_{J,E} S_{p,q} \,+\, ( p_{J,E})_{p,q}(\bvp^\ot)\right] \,,
\end{equation}
giving
\begin{equation}\label{eq-CE-fin}
 C_E \,\diag(\vze_{0,E})P_J \vec{e}_q =
\diag(\vze_{0,E}) P_J\,
\blb_{J,E} S \vec{e}_q + \diag(\vze_{0,E}) p_{J,E} \vec{e}_q
\end{equation}
where $ p_{J,E} \in [\Pp_{<|J|}(\bvp^\ot)]^{s\times s}$. 
Now for $J\in\Delta(m,\ZZ_+)$ define the vector spaces 
\beq
\VV_J:=\diag(\vze_{0,E})P_J \CC^s =\{\diag(\vze_{0,E})P_J \vec{v}\,:\,\vec{v}\in\CC^s\}\subset[\PE(m)]^s\subset\Hh_\infty^s\;.
\eeq
For $|J|\geq 1$ one finds $\VV_J\subset [\PE^{(0)}(m)]^s\subset (\Hh^{(0)}_\infty)^s$.
Given $u\in\ZZ_+$ define
\beq
\VV_u:=\bigoplus_{J:1\leq |J|\leq u} \VV_J \;.
\eeq
Note that $\Hh^s$ can be written as the direct vector sum $\Hh^s=\VV_0 \oplus \big(\Hh^{(0)}\big)^s$.
Using Lemma~\ref{lem-op-T}~(iii) one finds that $C_E$ leaves $\big(\Hh^{(0)}\big)^s$ invariant. Hence, one obtains
for $|J|\geq 1$ that $p_{J,E} (0)=0$ implying
$\diag(\vze_{0,E})\,p_{J,E}\, \vec e_q \,\in\,\VV_{|J|-1}\subset\VV_{|J|}$.
Therefore, $C_E$ leaves each of the spaces $\VV_u$ invariant. The spaces $\VV_u$ are nested, $\VV_u\subset\VV_{u+1}$ and 
using the basis $\diag(\vze_{0,E})P_J \vec e_{q}$ for $1\leq |J|\leq u$ and $q=1,\ldots,s$, ordered first by $J$ and then by $q$,
we see from \eqref{eq-CE-fin} that $C_E$ restricted to $\VV_u$ is represented by an upper block-triangular matrix consisting of
$s\times s$ matrix blocks. Moreover, the $s\times s$ blocks along the diagonal are given by the matrices
$\blb_{J,E} S$. Therefore, the eigenvalues of $C_E$ restricted to $\VV_u$ are exactly the eigenvalues of
the matrices $\blb_{J,E}S$ for $1\leq |J|\leq u$, i.e. $\sigma(C_E|_{\VV_u})=\bigcup_{1\leq|J|\leq u} \sigma(\blb_{J,E}S)$.
Using Lemma~\ref{lem-PEB} one obtains that $\bigcup_{u\geq 1} \VV_u$ is dense in $\big(\Hh^{(0)}\big)^s$.
As $C_E^2$ is compact, \eqref{eq-lbj} follows from Proposition~\ref{prop-appendix}.

 For part (iv) note that $\sigma_{\Hh^{(0)}_\infty}(C_E)\subset \sigma_{\Hh^{(0)}}(C_E)$ by compactness of
$C_E^2$ in $\Hh_\infty$. Equality follows as one finds eigenfunctions corresponding to the eigenvalues of $\blb_{J,E} S$ in
$\VV_{|J|}$ which is a subspace of $\Hh_\infty$.
\hfill $\Box$

\vspace{.2cm}

Similarly to above, define 
\begin{equation}
\vec\chi(\bvp_+^\ot,\bvp_-^\ot):=\left(\epsilon(\bvp_+^\ot)\epsilon(\bvp_-^\ot),\ldots,
\epsilon(\bvp_+^\ot)\epsilon(\bvp_-^\ot)\right)^\top\,\in\,[\PE(m)\otimes\PE(m)]^s\;
\end{equation}
and define the map $Q: \RR\times\RR\times [0,\infty)\times \big(\Kk^{(0)}_\infty \big)^s \to \big(\Kk^{(0)}_\infty \big)^s$ by
\begin{equation}
 Q(\lb,E,\eta,\vec g):=\Tt \Bb_{\lb,z} \left((\vec\chi +\vec g\,)^S\right)- \vec\chi\;.
\end{equation}
By Proposition~\ref{xita} this is a continuous map and using the definition of $\Tt$ and Lemma~\ref{lem-op-T}~(iii)
one obtains that indeed
$Q(\lb,E,\eta,\vec g)\in\big(\Kk^{(0)}_\infty \big)^s$
for $\vec g\in \big(\Hh^{(0)}_\infty \big)^s$. Moreover, one finds
\begin{equation}\label{eq-fxp-xi}
\vxi_{\lb,z} -\vec\chi \in \big(\Kk^{(0)}_\infty \big)^s \qtx{and}
Q(\lb,E,\eta,\vxi_{\lb,z} -\vec\chi)\;=\;\vxi_{\lb,z}-\vec\chi\;.
\end{equation}

\begin{lemma} \label{lem-DQ}
 {\rm (i)} The map $Q$ is Frechet-differentiable w.r.t. $\vec{g}\in\big(\Kk^{(0)}_\infty\big)^s$. The derivative $Q_{\vec g}$ is a bounded linear operator on $\big(\Kk^{(0)}_\infty\big)^s$ and extends naturally to a bounded operator on $\big(\Kk^{(0)}\big)^s$. \\
{\rm (ii)} For $E\in \hat I_{A,S}$ let $\Cc_E=Q_{\vec{g}}(0,E,0,\vxi_{0,E}-\vec\chi)$, then $\Cc_E^2$ is a compact operator on $\big(\Kk^{(0)}\big)^s$ and $\big(\Kk_\infty^{(0)}\big)^s$.\\
{\rm (iii)} The spectrum of $\Cc_E$ as an operator on the Hilbert space $\Kk^s$  is given by the eigenvalues of the matrices $\blb_{J,E}\blb^*_{J',E} S$ and the accumulation point $0$, i.e.
denoting the spectrum of $\Cc_E$ on $\big(\Kk^{(0)}\big)^s$ by $\sigma_{\Kk^{(0)}}(\Cc_E)$ one finds
\begin{equation}\label{eq-spectrum-Cc_E}
 \sigma_{\Kk^{(0)}}(\Cc_E)=\bigcup_{\substack{J,J'\in\Delta(m,\ZZ_+)\\|J|+|J'|\geq 1}} \sigma(\blb_{J,E}\blb^*_{J',E}S) \cup\{0\}\,
\end{equation}
where $\blb_{J,E}, \blb_{J',E}$ are the matrices as defined in \eqref{eq-def-blb}.
By definition of $\hat I_{A,S}$ one finds for $E\in \hat I_{A,S}$
\beq
1 \,\not\in\,\sigma_{\Kk^{(0)}}(\Cc_E)\,\label{eq-spectrum-CcE}\,.
\eeq
\noindent {\rm (iv)} The spectrum of  $\Cc_E$ as an operator on $\big(\Kk^{(0)}_\infty\big)^s$, denoted by 
$\sigma_{\Kk^{(0)}_\infty}(\Cc_E)$, is the same
as  its spectrum as an operator on $\Kk^s$:  
\beq \label{eq-spectrum-CCEinfty}
\sigma_{\Kk^{(0)}_\infty}(\Cc_E) = \sigma_{\Kk^{(0)}}(\Cc_E).
\eeq

\end{lemma}

\noindent {\bf Proof.} The proof is completely analogous to the one for Lemma~\ref{lem-DF}.
For (i) and (ii) note that the Frechet derivative of $Q$ is given by
\begin{equation}
 Q_{\vec{g}}\,=\,\Tt\cB_{\lambda,z} \Mm(\diag((\vec\chi+\vec{g})^S)\,S\,\diag((\vec\chi+\vec{g})^{-\one}))\;.
\end{equation}
For (iii) one starts with the identity
\begin{align}\label{eq-Cc}
 &\left[\Cc_E e^{\frac{i}2 \Tr(tM_+\bvp_+^\ot-t'M_-\bvp_-^\ot)}\diag(\vxi_{E,0}) \vec{e}_q\right]_p \\ \notag
&\qquad = e^{\frac{i}2 \Tr\left(\bvp_+^\ot(A-E-\sum_{r} S_{p,r}A^{(r)}_E-tM_+)^{-1}\right)}
e^{-\frac i2 \Tr\left(\bvp_-^\ot(A-E-\sum_{r} \bar S_{p,r} \bar A^{(r)}_E-t'M_- )^{-1}\right)}\;.
\end{align}
Analogously to above let us define
\begin{equation}
 P_J^+(\bvp_+^\ot,\bvp_-^\ot)=P_J(\bvp_+^\ot)\;,\qquad
 P_J^-(\bvp_+^\ot,\bvp_-^\ot)=P_J(\bvp_-^\ot)\;
\end{equation}
for $J,J'\in\Delta(m,\ZZ_+)$.
Performing a multi-variable Taylor expansion with respect to $t, t'$ in \eqref{eq-Cc}, 
comparing the terms and following similar steps as above
one obtains
\begin{align}
 \Cc_E \,\diag(\vxi_{E,0})P^+_J P_{J'}^-\vec{e}_q=
\diag(\vxi_{0,E})\left[P_J^+ P_{J'}^- \blb_{J,E} \blb^*_{J',E} S+ p_{J,J',E}  \right] \vec{e}_q\;,
\end{align}
where ${p}_{J,J',E}\in[\Pp_{<|J|+|J'|}(\bvp_+^\ot,\bvp_-^\ot)]^{s\times s}$. Setting
\begin{equation}
 \WW_{J,J'}=\diag(\vxi_{E,0})P^+_J P_{J'}^-\CC^s\;,\quad
\WW_u=\bigoplus_{J,J': 1\leq |J|+|J'|\leq u} \WW_{J,J'}
\end{equation}
one obtains that $\Cc_E$ leaves $\WW_u$ invariant and the restriction can be written as a block upper triangular matrix using $s\times s$ blocks.
The blocks along the diagonal are given by $\blb_{J,E}\blb_{J',E}^* S$ and hence
$$
\sigma(\Cc_E|_{\WW_u})=\bigcup_{J,J':1\leq|J|+|J'|\leq u} \sigma(\blb_{J,E} \blb_{J',E}^* S)\;.
$$
Again, using Lemma~\ref{lem-PEB} one realizes that $\bigcup_{u\geq 1} \WW_u$ is dense in $\big(\Kk^{(0)}\big)^s$, hence 
\eqref{eq-spectrum-Cc_E} follows from Proposition~\ref{prop-appendix}.

(iv) follows by the same arguments as in Lemma~\ref{lem-DF}.
 \hfill $\Box$

%%%%%%%%%%%%%%%%%%%%%%%%%%%%%%%%%%%%%%%%%%%%%%%%%%%%%%%%%%%%%%%%%%%%%%%%%%%%%%%%%%%%%%%%
%%%%%%%%%%%%%%%%%%%%%%%%%%%%%%%%%%%%%%%%%%%%%%%%%%%%%%%%%%%%%%%%%%%%%%%%%%%%%%%%%%%%%%%%
%%%%%%%%%%%%%%%%%%%%%%%%%%%%%%%%%%%%%%%%%%%%%%%%%%%%%%%%%%%%%%%%%%%%%%%%%%%%%%%%%%%%%%%%
%%%%%%%%%%%%%%%%%%%%%%%%%%%%%%%%%%%%%%%%%%%%%%%%%%%%%%%%%%%%%%%%%%%%%%%%%%%%%%%%%%%%%%%%
   
\section{Proof of the main Theorems} \label{sec-proofs}

%%%%%%%%%%%%%%%%%%%%%%%%%%%%%%%%%%%%%%%%%%%%%%%%%%%%%%%%%%%%%%%%%%%%%%%%%%%%%%%%%%%%%%%%
%%%%%%%%%%%%%%%%%%%%%%%%%%%%%%%%%%%%%%%%%%%%%%%%%%%%%%%%%%%%%%%%%%%%%%%%%%%%%%%%%%%%%%%%
%%%%%%%%%%%%%%%%%%%%%%%%%%%%%%%%%%%%%%%%%%%%%%%%%%%%%%%%%%%%%%%%%%%%%%%%%%%%%%%%%%%%%%%%
%%%%%%%%%%%%%%%%%%%%%%%%%%%%%%%%%%%%%%%%%%%%%%%%%%%%%%%%%%%%%%%%%%%%%%%%%%%%%%%%%%%%%%%%

The most important ingredient is the following proposition.

\begin{prop}  \label{xize}
For any $E\in \hat I_{A,S}$ there exist $\lb_E > 0$ and $\varepsilon_E >0$, 
 such that the maps 
\begin{equation}
(\lb, E',\eta) \in (-\lb_E,\lb_E)\times (E - \ve_E,E + \ve_E) \times (0, \infty) \;\mapsto\;
 \vze_{\lb,E'+i\eta} \in {\Hh}^s_\infty  \label{xizeze}
\end{equation}
and 
\begin{equation}
(\lb, E',\eta) \in (-\lb_E,\lb_E)\times (E - \ve_E,E + \ve_E) \times (0, \infty) \;\mapsto\;
 \vxi_{\lb,E'+i\eta} \in {\Kk}^s_\infty  \label{xizexi}
\end{equation}
have  continuous extensions to $ (-\lb_E,\lb_E)\times (E - \ve_E,E + \ve_E) \times [0, \infty)$
satisfying \eqref{eq-fpz} and \eqref{eq-fpx}, respectively.   
\end{prop}

\begin{proof}
We will use the Implicit Function Theorem on Banach Spaces 
as stated in \cite[Appendix~B]{Kl6}, a rewriting of \cite[Theorem~2.7.2]{N}.
Consider the function $\hat F(\lb,E,\eta,\vec f)= F(\lb, E,\eta,\vec f)-\vec f$.
By Lemma~\ref{lem-DF}, especially \eqref{eq-spectrum-CE}, the Frechet derivative $\hat F_{\vec f}(0,E,0,\vze_{0,E}-\vec\epsilon)=C_E-\one$
has no zero eigenvalue on $\big(\Hh^{(0)}_\infty\big)^s$ if $E\in \hat I_{A,S}$. 
Hence, the Implicit Function Theorem can be applied.
As a consequence, for each  $E \in \hat I_{A,S}$
there exist $\lb_E > 0$, $\varepsilon_E >0$, $\eta_E >0$ and $\de_E >0$, such that for each 
$$  (\lb, E',\eta) \in (-\lb_E,\lb_E)\times (E - \ve_E,E + \ve_E) \times [0, \eta_E) $$
there is a unique $\vec\omega_{\lb,E',\eta} \in \big(\Hh^{(0)}_\infty\big)^s$ with 
$\| \vec\omega_{\lb,E',\eta}- (\vze_{0,E}-\vec\epsilon) \|_{ {\Hh}^s_\infty}< \de_E$, 
such that we have  $  F(\lb, E',\eta,\vec\omega_{\lb,E',\eta}) =\vec\omega_{\lb,E',\eta}$. 
Moreover, the map
$$
 (\lb, E',\eta) \in (-\lb_E,\lb_E)\times (E - \ve_E,E + \ve_E) \times [0, \eta_E) \;\longrightarrow\;
 \vec\omega_{\lb,E',\eta} \in {\Hh}^s_\infty
$$
is continuous.
To obtain the statement for the map \eqref{xizeze} it is left to show
\begin{equation} \label{eq-om-ze}
\vze_{\lb,E' + i\eta}-\vec\epsilon = \vec\omega_{\lb,E',\eta}
\end{equation}
 for all $(\lb, E',\eta) \in (-\lb_E,\lb_E)\times (E - \ve_E,E + \ve_E) \times [0, \eta_E)$.
But it follows  from Proposition~\ref{zeta} that  $\vze_{\lb,E'+i\eta}-\vec\epsilon \in \big(\Hh^{(0)}_\infty\big)^ s$ is continuous
on $\left(\{0\} \times  \{E'\}\times  [0, \eta_1]\right) \cup  
 \left(\RR \times\RR \times [\eta_1, \infty) \right)$, for any $\eta_1 >0$, and it satisfies 
\eqref{eq-fxp}.  Thus \eqref{eq-om-ze} follows from the uniqueness in the Implicit Function Theorem.  

Using Proposition~\ref{xita}, \eqref{eq-fxp-xi} and Lemma~\ref{lem-DQ}, the proof for the map in \eqref{xizexi} is 
completely analogous.
\end{proof}

\begin{remark}
 The use of the spaces $\Hh^{(0)}$ and $\Kk^{(0)}$ lies in equations
\eqref{eq-lbj}, \eqref{eq-spectrum-CE}, \eqref{eq-spectrum-Cc_E} and \eqref{eq-spectrum-CcE} which are crucial for the proposition above. 
If one works with the spaces $\Hh$ and $\Kk$ instead, this would mean that the spaces $\VV_{{\mathbf 0}}$ and $\WW_{{\mathbf 0},{\mathbf 0}}$ have to be added
and the calculated Frechet derivatives $F_{\vec{f}}$ and $Q_{\vec g}$ on $\Hh$ and $\Kk$ would 
get additional eigenvalues which are equal to the eigenvalues of $S$.
For this reason, we would have to demand $\det(S-\one)\neq 0$ in order to use the Implicit Function Theorem as done above.
This would artificially rule out some of the substitution trees, e.g. the ones associated to 
$S=\left(\begin{smallmatrix} 4 & 3 \\ 2 & 3
\end{smallmatrix} \right)$ which satisfies {\rm (S1)}, {\rm (S2)} and {\rm (S3)}.
\end{remark}

\begin{coro}\label{cor-xize}
For $\im(z)>0$ and neighboring sites $x,y\in\IS^{(q)}$ one obtains the following.
\begin{align}
& \ze^{(x|y)}_{\lb,z} \in {\Hh}_\infty,\qtx{if $x\neq 0_q\;$ and} \ze^{(x|y)}_{\lb,z}=\ze^{(0_q|y)}_{\lb,z} \in \Hh\;,
\quad\text{if $x=0_q$} \label{eq-st1}\\
& \xi^{(x|y)}_{\lb,z} \in {\Kk}_\infty,\qtx{if $x\neq 0_q\;$ and}
\xi^{(x|y)}_{\lb,z}=\xi^{(0_q|y)}_{\lb,z} \in \Kk\;,\quad\text{if $x=0_q$}\;.\label{eq-st2}
\end{align}
Moreover, there is a neighborhood $U$ of $\{0\} \times \hat I_{A,S}$ in $\RR^2$ such that
for all $x,y\in \IS^{(q)}$ with $d(x,y)=1$, the maps
\begin{align}
 & (\lb,E,\eta)\in\RR\times\RR\times(0,\infty)\;\mapsto\;\ze^{(x|y)}_{\lb,z} \in {\Hh}_\infty \;,\quad 
\text{{\rm (}$\Hh$, respectively, if $x=0_q${\rm )}}\label{eq-cext1}\\
& (\lb,E,\eta)\in\RR\times\RR\times(0,\infty)\;\mapsto\;\xi^{(x|y)}_{\lb,z} \in {\Kk}_\infty\;,\quad 
\text{{\rm (}$\Kk$, respectively, if $x=0_q${\rm )}} \label{eq-cext2}
\end{align}
have continuous extensions to $(\lb,E,\eta)\in U\times[0,\infty)$.
\end{coro}
\begin{proof}
If $d(x,y)=1$ and $x$ is a child of $y$ (i.e. $d(x,0_q)=1+d(y,0_q)$), 
then $\IS^{(x|y)}$ is equivalent to $\IS^{(p)}$ if the label of $x$ is $p$, $l(x)=p$, and therefore 
$\ze^{(x|y)}_{\lb,z}=\ze^{(p)}_{\lb,z}$ and $\xi^{(x|y)}_{\lb,z}=\xi^{(p)}_{\lb,z}$.
Together with Proposition~\ref{xize} this implies all the statements if $x$ is a child of $y$.
For the other statements,
let us first consider the case $x=0_q$. From Appendix~\ref{app-gr}, equation \eqref{eq-ze-rec-gen} we find
\begin{align}
 \ze^{(0_q|y)}_{\lb,z} &=
TB_{\lb,z} \bigg(\prod\limits_{\substack{y': d(0_q,y')=1\\ y'\neq y}} 
\ze^{(y'|0_q)}_{\lb,z}
\bigg)\;.
\end{align}
As the $y'$ above are all children of the root $0_q$, $\ze^{(y'|0_q)}_{\lb,z}\in\Hh_\infty$.
By assumption (S1), $0_q$ has at least two children so the product over $y'$ is not empty.
If it has one factor, then the product is an element of $\Hh_\infty\subset\Hh$. 
If it has more than one factor we use Lemma~\ref{lemma-prod} to get that the product is in $\Hh_1\subset\Hh$.
As $T$ and $B_{\lb,z}$ are bounded operators on $\Hh$
one gets $\ze^{(0_q|y)}_{\lb,z}\in\Hh$.

Now we will use an induction argument over $d(x,0_q)$. Let $x'$ be the parent of $x$ and $x$ be the parent of $y$. Assume we already know $\ze^{(x'|x)}_{\lb,z}\in \Hh$ and that it extends continuously in $\Hh$.
(We need to take the space $\Hh$ in the assumption in the first step, when $x'=0_q$).

Then
\begin{equation}
\ze^{(x|y)}_{\lb,z}=TB_{\lb,z} \bigg(\prod\limits_{\substack{y': d(x,y')=1\\ y'\neq y}}
\ze^{(y'|x)}_{\lb,z}
\bigg)\;.
\end{equation}
Except for $y'=x'$, all neighbors $y'$ of $x$ will be children of $x$ and hence 
$\ze^{(y'|x)}_{\lb,z}$ extends continuously in $\Hh_\infty$. 
As $x$ has at least two children by assumption (S1), there is at least one which is not $y$. Using the induction assumption, Lemma~\ref{lemma-prod} 
and the fact that $TB_{\lb,z}$ maps continuously from $\Hh_1$ to $\Hh_\infty$, we obtain that 
$\ze^{(x|y)}_{\lb,z}$ extends continuously in $\Hh_\infty$.

With similar arguments, based on the relations \eqref{eq-xi-rec-gen}, one obtains the statements for the functions $\xi^{(x|y)}_{\lb,z}$.
\end{proof}
\begin{remark}
 If the root $0_q$ has 3 children, then $\ze^{(0_q|y)}\in\Hh_\infty$ and $\xi^{(0_q|y)}_{\lb,z}\in\Kk_\infty$.
\end{remark}

Analogously to \cite{KS}, Theorem~\ref{main2} follows directly from the Corollary above.
More precisely, integrating out the Grassmann variables of the supersymmetric integral identities \eqref{eq-EG-x} and \eqref{eq-EGG-x},
we obtain 
\begin{align}\label{eq-EG-x-D}
\E\big(G^{[x]}_{\lb}\, (z)\big) &= -i\int {\mathbf D}B_{\lb,z} \left( \prod_{y:d(x,y)=1} \ze^{(y|x)}_{\lb,z}(\bvp^\ot)\right)\,
d^{2mn} \bvp\;, \\
\E\left(\left| G^{[x]}_{\lb}\, (z)\right|^2\right) &= 
\int {\mathbf D}^{(-)} {\mathbf D}^{(+)} \Bb_{\lb,z} \left( \!\prod_{y:d(x,y)=1}\!\!\! \xi^{(y|x)}_{\lb,z}(\bvp_+^\ot, \bvp_-^\ot)\right)
\,d^{2mn} \bvp_+\,d^{2mn} \bvp_- \;, \label{eq-EGG-x-D}
\end{align}
where ${\mathbf D}$ is an $m\times m$ matrix of differential operators with entries
\begin{equation}\label{eq-def-bfD}
 {\mathbf D} :=\left( \tfrac{(-1)^{mn+j+k}}{2\pi^{mn}}\,D_{\I,\I}^{n-1}\,D_{\I\setminus\{k\},\I\setminus\{j\}}\right)_{j,k\in\I}\;.
\end{equation}
where $D_{\I,\I}$ and $D_{\I\setminus\{k\},\I\setminus\{j\}}$ are defined as in \eqref{eq-def-Daa}.
${\mathbf D}^{(\pm)}$ represent the matrix-operator ${\mathbf D}$ acting with respect to $\bvp_\pm^\ot$ and
${\mathbf D}^{(-)}{\mathbf D}^{(+)}$ has to be understood as a matrix product.
By assumption {\rm (S1)}, the products inside the integrals in \eqref{eq-EG-x-D} and \eqref{eq-EGG-x-D} contain at least two factors,
by Corollary~\ref{cor-xize}, all factors are in $\Hh_\infty$, or $\Kk_\infty$, respectively, except for possibly one, 
which is in $\Hh$, or $\Kk$, respectively,
and they depend continuously on $(\lambda,E,\eta) \in U\times [0,\infty)$.
By Lemma~\ref{lemma-prod}, the product is in $\Hh_1$, or $\Kk_1$, respectively, and by Dominated Convergence, Theorem~\ref{main2} now follows.

\vspace{.2cm}

Now we can prove Theorem~\ref{main}. Let us first start with part (ii).
For $\lambda=0,\,E\in \hat I_{A,S}$ one has $\im (G_0^{(q)}(E))>0$. Therefore, Theorem~\ref{main2} implies
that there is an open neighborhood $U$ of $\{0\}\times \hat I_{A,S}$ such that
for all $x\in\IS$, $\E (G^{[x]}_{\lb,z}(E+i\eta))$ and $\E(|G_\lambda^{[x]}(E+i\eta)|^2)$ extend continuously to $(\lb,E,\eta)\in U\times[0,\infty)$
and such that for all $q\in\{1,\ldots,s\}$ one has
\begin{equation}
\lim_{\eta\downarrow 0} \im\,\E(G_\lambda^{(q)}(E+i\eta)) > 0\qtx{for} (\lb,E)\in U\;. 
\end{equation}
The latter can be achieved by continuity and possibly shrinking the neighborhood $U$ as in
Theorem~\ref{main2}, since $\lim_{\eta\downarrow 0} \im\, G_0^{(q)}(E+i\eta)>0$ for $E\in \hat I_{A,S}$. 
Theorem~\ref{main}~(ii) follows as
$\E(G_\lambda^{[x]}(E+i\eta))$ is the Stieltjes transform of $\E(\mu_{x})$.

Using Fubini's Theorem, Fatou's Lemma and the continuous extension of
$z\mapsto \E\big(\big|G^{[x]}_\lb(z) \big|^2\big)$ to $z\in U_\lb$, one obtains for an interval $[a,b]\subset U_\lb$
\begin{align}
&\E \left( \liminf_{\eta \downarrow 0}\int_{a}^b \Tr\left(\left|G^{[x]}_{\lb}\, (E +i\eta)\right|^2\right)\,dE \right)\\
& \qquad\qquad
\leq\;\; \liminf_{\eta \downarrow 0} \int_{a}^b \E \left(\Tr\left(\left|G^{[x]}_{\lb}\, (E +i\eta)\right|^2\right)  \right)\,dE\;\;<\;\; \infty\;. \notag
\end{align}
Thus, 
\beq\label{limprobone}
\liminf\limits_{\eta \downarrow 0}\int\limits_{a}^b \Tr\left(\left|G^{[x]}_{\lb}\, (E +i\eta)\right|^2\right)\,dE \,<\,\infty \quad
\text{with probability one}. 
\eeq
%%%

Clearly, $G_\lb^{[x]}(z)$ is the Stieltjes transform of $\mu_{x}$. Therefore,
similarly as in \cite{KS}, based on \cite[Theorem 4.1]{Kl6} (or \cite[Theorem 2.6]{Kel}) one obtains from \eqref{limprobone}
that, almost surely, $\mu_{x}$ is absolutely continuous with respect to the Lebesgue measure in 
$(a,b)\subset U_\lambda$ and the density is a matrix valued $L^2$ function. As there are countably many vertices $x\in\IS$ 
and the open set $U_\lb$ can be obtained as countable union of intervals $(a,b)$ satisfying
$[a,b]\subset U_\lb$, we obtain that with probability one,  all measures $\mu_x$ are absolutely continuous in $U_\lb$. Thus, we finally proved part (i) of Theorem~\ref{main}.

\appendix

\section{Spectrum of operators with compact power}

\begin{prop} \label{prop-appendix}
Let $C$ be an operator on a Hilbert space $\Hh$ such that a power $C^j$ is compact.
Let moreover $\VV_u$ for $u\in\ZZ_+$ be finite dimensional subspaces such that
$C$ leaves $\VV_u$ invariant and such that $\bigcup_{u\in\ZZ_+} \VV_u$ is dense in $\Hh$.
Then one has
\begin{equation}
 \sigma(C)=\bigcup_{u\in\ZZ_+} \sigma(C|_{\VV_u}) \cup \{0\}\;,
\end{equation}
where $\sigma(C|_{\VV_u})$ denotes the set of eigenvalues of the restriction of $C$ to the finite
dimensional space $\VV_u$.
\end{prop}

\begin{proof}
As $C$ leaves $\VV_u$ invariant and as $0\in\sigma(C)$ since $C^j$ is compact, the inclusion '$\supset$' is trivial.
So let $c\not\in\sigma(C|_{\VV_u})$ for any $u\in\ZZ_+$ and $c\neq 0$, we have to show $c\not\in\sigma(C)$.
Now, $(C-c)|_{\VV_u}$ is invertible on $\VV_u$ and hence $(C-c)(\VV_u)=\VV_u$.
Therefore, the range of $C-c$ includes $\bigcup_u \VV_u$ and is dense in $\Hh$, and hence
$$
\ker(C^*-\bar c)= [(C-c)(\Hh)]^\perp = \{0\},.
$$
Thus, $\bar c\neq 0$ is not an eigenvalue of $C^*$. Since $(C^*)^j$ is compact, this means $\bar c$ is not in the spectrum of
$C^*$, and $C^*-\bar c=(C-c)^*$ is invertible. But this means $C-c$ is invertible and hence $c\not\in\sigma(C)$.
\end{proof}

\section{Supersymmetric methods}
\label{sec-super}

%%%%%%%%%%%%%%%%%%%%%%%%%%%%%%%%%%%%%%%%%%%%%%%%%%%%%%%%%%%%%%%%%%%%%%%%%%%%%%%%%%%%%%%%
%%%%%%%%%%%%%%%%%%%%%%%%%%%%%%%%%%%%%%%%%%%%%%%%%%%%%%%%%%%%%%%%%%%%%%%%%%%%%%%%%%%%%%%%
%%%%%%%%%%%%%%%%%%%%%%%%%%%%%%%%%%%%%%%%%%%%%%%%%%%%%%%%%%%%%%%%%%%%%%%%%%%%%%%%%%%%%%%%
%%%%%%%%%%%%%%%%%%%%%%%%%%%%%%%%%%%%%%%%%%%%%%%%%%%%%%%%%%%%%%%%%%%%%%%%%%%%%%%%%%%%%%%%

For the readers convenience we briefly give the supersymmetric background for the definitions and identities in Sections~\ref{sec-ban}
and \ref{sec-fxp}. We use the notations
as introduced in \cite{KS,KS2} which give a more conceptual introduction.

Given an alphabet $\Aa$ (set of symbols) let $\Lambda(\Aa)$ denote the Grassmann algebra generated by the symbols in $\Aa$.
$\Lambda(\Aa)$ is given by the free algebra over the alphabet $\Aa$ modulo the anti-commutation relations $ab+ba=0$ for $a,b\in\Aa$. 
The elements of $\Aa$ will be called independent Grassmann variables. The set of linear combinations of elements of $\Aa$ are the
so called one forms in $\Lambda(\Aa)$.
Clearly, if $\Aa\subset\Bb$, then $\Lambda(\Aa)\subset\Lambda(\Bb)$ 
can be naturally considered as a sub-algebra. Hence, one can always introduce a new Grassmann variable which is independent 
to all Grassmann variables already used.

Let $\psi_{k,\ell}, \overline{\psi}_{k,\ell}$ for $k=1,\ldots,m$ and $\ell=1,\ldots,n$ be $2mn$ independent Grassmann variables
(letters in the alphabet). Together with an ordinary variable $\varphi_{k,\ell}\in\RR^2$ one has $mn$ so called supervariables 
$\phi_{k,\ell}=(\varphi_{k,\ell},\overline{\psi}_{k,\ell},\psi_{k,\ell})$.
The collection $\BPhi=(\phi_{k,\ell})_{k,\ell}$ will be called a $m \times n$ supermatrix.
Its ordinary variables part, $\bvp=(\varphi_{k,\ell})_{k,\ell}$ with entries in $\RR^{2}$, as well as the Grassmann variables part
$\BPsi=(\overline\psi_{k,\ell},\psi_{k,\ell})$ will be considered as $m\times 2n$ matrices and one may write
$\BPhi=(\bvp,\BPsi)$. The set of all $m\times n$ supermatrices whose Grassmann variables are one forms of a Grassmann algebra
$\Lambda(\Aa)$ will be denoted by $\Ll_{m,n}(\Aa)$.
Two supermatrices $\BPhi,\BPhi'\in\Ll_{m,n}(\Aa)$ will be called independent, if all one-forms are linearly independent.

For an $m\times n$ supermatrix $\BPhi\in\Ll_{m,n}(\Aa)$ and an $m\times m$ matrix $B$, we define
$B\BPhi \in\Ll_{m,n}(\Aa)$ by
\begin{equation}\label{eq-BBPhi}
\BPhi'=B\BPhi\quad:\Leftrightarrow\quad \phi'_{j,\ell} = \sum_{k=1}^m B_{j,k} \phi_{k,\ell}=
\sum_{k=1}^m\left(B_{j,k} \varphi_{k,\ell}, B_{j,k}\overline \psi_{k,\ell}, B_{j,k} \psi_{k,\ell} \right)\;.
\end{equation}
Moreover, for supervariables $\phi_1=(\varphi_1,\overline{\psi}_1,\psi_1)$, $\phi_2=(\varphi_2,\overline{\psi}_2,\psi_2)$
we define
\begin{equation}
 \label{eq-dot-supervec}
\phi_1 \cdot \phi_2 : = \varphi_1\cdot\varphi_2+\tfrac12(\overline{\psi}_1 \psi_2 + \overline\psi_2 \psi_1)\;,
\end{equation}
and for supermatrices $\BPhi'$, $\BPhi$, let 
\begin{equation}\label{eq-dot-supermatrix}
 \BPhi' \cdot \BPhi : = 
\sum_{k=1}^m \sum_{\ell=1}^n \phi'_{k,\ell}\cdot\phi_{k,\ell}
\qtx{and}
\BPhi'\cdot B \BPhi := \sum_{j,k=1}^m \sum_{\ell=1}^n B_{j,k} \phi'_{j,\ell} \cdot \phi_{k,\ell}\;.
\end{equation}
Furthermore, we define the $m\times m$ matrix $\BPhi^\ot$ with entries in $\Lambda(\BPsi)$, by
\begin{equation}\label{eq-supermatrix-tensor}
(\BPhi^{\od 2})_{j,k} :=  \sum_{\ell=1}^n \phi_{j,\ell} \cdot \phi_{k,\ell} \;.
\end{equation}
Note that
\begin{equation}
\BPhi\cdot B \BPhi = \Tr(B \BPhi^\ot)\;.
\end{equation}
For a supermatrix $\BPhi=(\bvp,\BPsi)$, 
one finds
\beq \label{eq-BPhi2}
\BPhi^\ot = \bvp^\ot\,+\,\BPsi^\ot\,, \qtx{with}
\bvp^\ot:=\bvp\bvp^\top\;\text{and}\;  \BPsi^\ot:=\BPsi J \BPsi^\top\; ,
\eeq
where $J$ is the $2n\times 2n$ matrix consisting of $n$ blocks 
$\left[\begin{smallmatrix} \,0 & \frac12 \\ -\frac12 & 0 \end{smallmatrix}\right]$ along the diagonal.
Given a matrix $B\in \CC^{m\times m}$ and $\bvp', \bvp\in \RR^{m\times 2n}$, we write
\begin{equation}\label{eq-supermatrix-prod25}
\bvp'\cdot B \bvp := \sum_{j,k,\ell} B_{j,k} \varphi'_{j,\ell} \cdot \varphi_{k,\ell} = \Tr((\bvp')^\top B \bvp )\;\in\;\CC\; .
\end{equation}

%%%%%%%%%%%%%%%%%%%%%%%%%%%%%%%%%%%%%%%%%%%%%%%%%%%%%%%%%%%%%%%%%%%%%%%%%%%%%%%%%%%%%%%%%%%%%%%%%%%%%%%%%%%%%%%%%%%%%%%%%%

Next let us recall the convenient notation for Grassmann monomials as in \cite{KS2}.
For convenience, let $\I=\{1,\ldots,m\}$ and  denote the set of subsets of $\I$ by
$\PP(\I)$.
Given $(\bar a, a)\in  (\PP(\I))^2=\PP(\I)\times\PP(\I)$
and $\bar a=\{\bar k_1,\ldots, \bar k_c\},\;a=\{k_1,\ldots, k_d\}$, both ordered ({ i.e.}, $\bar k_i < \bar k_j$
and $k_i<k_j$ if $i<j$),  we set
\begin{equation}\label{Psiaal}
\Psi_{\bar a, a,\ell} :=  \left(\prod_{j=1}^{|\bar a|}\;\overline{\psi}_{\bar k_j, \ell}\,\right)\left(\prod_{j=1}^{|a|}
\psi_{k_j,\ell}\,\right)\;,
\end{equation}
using the conventions $\prod_{j=1}^c \psi_j=\psi_1 \psi_2 \cdots\, \psi_c$ for non-commutative products and
$\prod_{j=1}^0 \psi_j = 1$.
In particular, $\Psi_{\emptyset,\emptyset,\ell}=1$.
Given a pair of $n$-tuples of subsets of $\I$, $(\bar\aaa,\aaa) \in (\PP(\I))^n\times(\PP(\I))^n $ with 
$\bar\aaa=(\bar{a}_1,\ldots,\bar{a}_n)$ and $\aaa=(a_1,\ldots,a_n)$, we set
\begin{equation} \label{eq-Psi-a-a}
 \Psi_{\bar\aaa,\aaa}:=
\prod_{\ell=1}^n \Psi_{\bar a_\ell,a_\ell, \ell}\;.
\end{equation}

For $f\in C_n^\infty(\Sym^+(m))$ and $\det(\bvp^\ot) \neq 0$, a formal Taylor expansion yields the definition
(cf. \cite[eq.(2.17)]{KS2} and \cite[eq.(2.26)]{KS})
\begin{equation}\label{eq-super-Taylor}
 f(\BPhi^{\od 2}) := 
\sum_{ (\bar\aaa, \aaa)\,\in\,\Pp^n}
D_{\bar\aaa,\aaa}\; f(\bvp^\ot)\;\sgn(\aaa) \Psi_{\bar\aaa,\aaa}\;.
\end{equation}
where $D_{\bar\aaa,\aaa}$ is defined as in \eqref{eq-def-Daa} and
\begin{equation}
 \sgn(\aaa)=\prod_{\ell=1}^n (-1)^{\frac{| a_\ell|(| a_\ell|-1)}{2}}\;,\quad\aaa\in(\PP(\I))^n\;.
\end{equation}

Consider the Grassmann algebra over some alphabet $\Aa$ and let $\psi\in \Aa$. 
Then, using the anti-commutation relations, $F\in \Lambda(\Aa)$ can be uniquely written as $F=F_0+F_1\psi$
with $F_0,F_1\in\Lambda(\Aa\setminus\{\psi\})$. The Berezin integral over $d\psi$ is defined 
as a linear map from $\Lambda(\Aa)$ to $\Lambda(\Aa\setminus\{\psi\})$ by
\begin{equation}\label{eq-def-intdpsi}
\int F\,d\psi=
\int F_0+F_1 \psi\,d\psi:= F_1 \;.
\end{equation}
Because of the anti-commutativity of $\psi,\bar \psi \in \Aa$ in $\Lambda(\Aa)$ one finds
\begin{equation}
 \int F\bar\psi\psi\, d\bar\psi\,d\psi\,=\,
\int -F \psi \bar\psi\,d\bar\psi\,d\psi=\int -F \psi\,d\psi =  - F\;.
\end{equation}

A superfunction $F(\BPhi)$ for $\BPhi=(\bvp,\BPsi)$ is a function mapping $\bvp\in\RR^{m\times 2n}$ 
to $\Lambda(\BPsi)$ and can be written as
\begin{equation}
F(\BPhi)= \sum_{(\bar\aaa,\aaa)\in(\PP(\I))^n\times(\PP(\I))^n} F_{\bar\aaa,\aaa}(\bvp) \Psi_{\bar\aaa,\aaa}\;,
\end{equation}
where $F_{\bar\aaa,\aaa}$ are functions from $\RR^{m\times 2n}$ to $\CC$.
We call $F$ smooth or integrable if all $F_{\bar\aaa,\aaa}$ are smooth or integrable, respectively.
If $F(\BPhi)$ is integrable, then we define
\begin{equation}
 \int F(\BPhi)\,D\BPhi\,=\,
\frac{1}{\pi^{mn}} \int F(\BPhi) \prod_{k=1}^m \prod_{\ell=1}^n d^2 \varphi_{k,\ell}\,d\overline{\psi}_{k,\ell}\,d\psi_{k,\ell}\;.
\end{equation}
A superfunction is called supersymmetric, if $F(\BPhi)=f(\BPhi^\ot)$.
An important integral equality is given in \cite[Proposition II.2.10]{KSp} which implies for a 
smooth, integrable, supersymmetric function $F(\BPhi)=f(\BPhi^\ot)$, that
\begin{equation}\label{eq-sup-id}
 \int F(\BPhi)\,D\BPhi =
\int f(\BPhi^\ot)\,D\BPhi=f(0)\;.
\end{equation}

Let $\BPhi,\BPhi'$ be independent super matrices, 
and let $A,B $ be two complex $m\times m$ matrices with $B=B^\top$ and $2\re(B)=B+B^*>0$. Then,
a special case of \eqref{eq-sup-id} together with some changes of variables shows
(cf. \cite[Theorem~III.1.1]{KSp})
\begin{equation}\label{eq-sup-id-1}
 \int e^{-(\BPhi+A\BPhi')\cdot B (\BPhi+A\BPhi')}\,D\BPhi = 1\;,
\end{equation}
where the exponential is defined by a formal Taylor expansion in the Grassmann variables.
Using $2\BPhi \cdot A\BPhi' - \BPhi \cdot B \BPhi =
-(\BPhi-B^{-1}A \BPhi' ) \cdot B (\BPhi-B^{-1}A \BPhi')+\BPhi' \cdot A^\top B^{-1} A \BPhi'$ one obtains
for $B=B^\top, \re(B)>0$ that
\begin{equation}\label{eq-sup-id-2}
 \int e^{2\BPhi \cdot A \BPhi'}\,e^{-\BPhi \cdot B \BPhi}\,D\BPhi\,=\,
e^{\BPhi' \cdot (A^\top B^{-1} A) \BPhi'}\;.
\end{equation}

%\noindent {\bf Remark.} For these identities it is not important that $n\geq\frac m2$, they hold for any $n\geq 1$.

\vspace{.2cm}

Next, let us define some algebraic operations on $(\PP(\I))^n$ which will give a Leibniz type formula.
Let $\aaa, \bbb \in (\PP(\I))^n$. 
If $ a_\ell\cap b_\ell=\emptyset $ for each $\ell=1,\ldots,n$, then we say $\aaa$ and $\bbb$ are {\em addable} and define
$\ccc=\aaa+\bbb \in (\PP(\I))^n$ by
$c_\ell= a_\ell\cup b_\ell$.
We say that $(\bar\aaa,\aaa)$ and $(\bar\bbb,\bbb)\in \PP(\I)^n \times \PP(\I)^n$ are addable if
$\bar\aaa+\bar\bbb$ and $\aaa+\bbb$ are defined by the notion above.
In this case we define $\sgn(\bar\aaa,\aaa,\bar\bbb,\bbb)\in\{-1,1\}$ by
\begin{equation} \label{eq-def-sgn4}
\Psi_{\bar\aaa,\aaa}\,\Psi_{\bar\bbb,\bbb} = 
\sgn(\bar\aaa,\aaa,\bar\bbb,\bbb)\; \Psi_{\bar\aaa+\bar\bbb,\aaa+\bbb}\;.
\end{equation}
Since the product of two supersymmetric functions is supersymmetric,  for all
$f, g\in {C^\infty_n(\Sym^+(m))}$ and all $(\bar\aaa,\aaa)\in\Pp^n$ we have
\begin{equation} \label{eq-Leibn}
D_{\bar\aaa,\aaa}\,(fg) = 
\sum_{\substack{ (\bar\bbb,\bbb),(\bar\bbb',\bbb')\in\Pp^n\\ \bar\bbb+\bar\bbb'=\bar\aaa\;,
\bbb+\bbb'=\aaa}}\;\frac{\sgn(\bbb)\sgn(\bbb')\sgn(\aaa)}{\sgn(\bar\bbb,\bbb,\bar\bbb',\bbb')}\,
D_{\bar\bbb,\bbb}\, g\,D_{\bar\bbb',\bbb'}\, f\;.
\end{equation}
Furthermore, let us introduce $\III=(\I,\ldots,\I)\in(\PP(\I))^n$ as the $n$-tuple of subsets of $\I$ where each entry is the full set $\I$,
and for $\aaa\in(\PP(\I))^n$ define $\aaa^\brc \in (\PP(\I))^n$ as the element which is addable to $\aaa$ and
$\aaa+\aaa^\brc=\III$.

\vspace{.2cm}

For $n\geq \frac{m}{2}$, $f\in\Ss_n(\Sym^+(m))$ the supersymmetric Fourier transform $Tf$ is defined by
\begin{equation}\label{eq-def-T}
 (Tf)((\BPhi')^{\ot}) = 
\int e^{\pm\imath \BPhi'\cdot\BPhi}\, f({\BPhi}^{\ot})\; D\BPhi\; ,
\end{equation}
$T$ maps $\Ss_n(\Sym^+(m))$ into itself.
Since $(-\BPhi)^\ot=\BPhi^\ot$ a change of variables shows that $Tf$ does not depend on the sign
in the first exponent, $ \pm\imath \BPhi'\cdot\BPhi$.
One finds  \cite[eq. (2.37)]{KS}
\begin{equation}\label{eq-T-parts}
D_{\bar\aaa,\aaa} (Tf) = 
\tfrac{2^{mn}}{4^{|\aaa|}}\;\sgn(\aaa,\bar\aaa)\;\Ff(D_{\aaa^\brc,\bar\aaa^\brc}\,f) \qtx{for all} (\bar\aaa,\aaa) \in \Pp^n\;,
\end{equation}
where 
\begin{equation}\label{eq-def-sgn2}
\sgn(\bar\aaa,\aaa) :=
(-1)^{mn}\frac{\sgn(\aaa)\sgn(\aaa^\brc)\sgn(\III)}{\sgn(\bar\aaa,\aaa,\bar\aaa^\brc,\aaa^\brc)}
\end{equation}
as in \cite{KS, KS2} and
$\Ff$ denotes the Fourier transform on $\RR^{m\times 2n}$; we abuse the notation by letting $\Ff f$ denote  
the function in $\Ss_n(\Sym^+(m))$ such that $(\Ff f)(\bvp^\ot)$ is the Fourier transform of the function  $F(\bvp)=f(\bvp^\ot)$.   

The equation \eqref{eq-T-parts} is also the main reason for the definition of the norms and the Banach spaces $\Hh$, $\Hh_p$.
In particular, one can read off directly that $T$ is an involution and unitary on $\Hh$.

\vspace{.2cm}

Considering the real variables part, $Tf((\bvp')^\ot)$, and integrating out the Grassmann
variables, \eqref{eq-def-T} leads to
\begin{align} \label{eq-def-T2}
Tf((\bvp')^\ot) &= \int e^{\pm i \bvp \cdot \bvp'} \sum_{(\bar\aaa,\aaa)\in\Pp^n} D_{\bar\aaa,\aaa} f(\bvp^\ot) \sgn(\aaa) \Psi_{\bar\aaa,\aaa} \;D\BPhi \\
%& = \frac{\sgn(\III)}{\pi^{mn}} \int e^{\pm i \bvp\cdot\bvp'} D_{\III,\III} f (\bvp^\ot)\;
%\sgn(\III) \left[\prod_{k=1}^m \prod_{\ell=1}^n 
% \bar\psi_{k,\ell} \psi_{k,\ell} \right]\, D\BPhi \notag \\
&= 
\frac{(-1)^{mn}}{\pi^{mn}} \int e^{\pm i \bvp\cdot\bvp'} D_{\III,\III} f (\bvp^\ot)\;
\prod_{k=1}^m \prod_{\ell=1}^n d^2 \varphi_{k,\ell}
\notag
\end{align}
which is exactly \eqref{eq-def-T0}.
Here, we used $\BPsi_{\III,\III} = \sgn(\III) \prod_{k=1}^n \prod_{\ell=1}^m \bar\psi_{k,\ell} \psi_{k\ell}$ and
$\prod_{k=1}^n \prod_{\ell=1}^m \int \bar\psi_{k,\ell} \psi_{k\ell}\; d\bar\psi_{k,\ell} \,d\psi_{k,\ell} = (-1)^{mn}$.
Plugging in $\bvp'=0$ and using \eqref{eq-sup-id} we obtain
\begin{equation}\label{eq-int-id0}
 Tf(0)=\frac{(-1)^{mn}}{\pi^{mn}}
\int e^{\pm i \bvp\cdot\bvp'} D_{\III,\III} f (\bvp^\ot)\;
\prod_{k=1}^m \prod_{\ell=1}^n d^2 \varphi_{k,\ell}\,=\, f(0)\;.
\end{equation}

Moreover, for $A=\frac{i}2 \one$, $B=-i\tilde B$ with $\tilde B=\tilde B^\top,\;\im(\tilde B)>0$ 
one has $A^\top B^{-1} A= -\frac{i}{4} \tilde B^{-1}$ and using \eqref{eq-sup-id-2},
one obtains
\begin{equation}
T(e^{i\Tr(\tilde B\BPhi'^\ot)})= \int e^{i\BPhi\cdot\BPhi'} e^{i \BPhi \cdot \tilde B \BPhi} \,D\BPhi
= e^{-\frac{i}{4} \BPhi' \cdot \tilde B^{-1} \BPhi'}\;.
\label{eq-T-id0}
\end{equation}
Considering the real variables part, this is equivalent to \eqref{eq-T-id}.

\section{Supersymmetric identities for the matrix Green's functions}\label{app-gr}

We will use the notations as introduced in Appendix~\ref{sec-super} and Section~\ref{sec-fxp}, esp. \eqref{tilde-G}, \eqref{tilde-G1}, \eqref{eq-zeta-g} and \eqref{xia-g}.
Using the supersymmetric replica trick \cite{Ber,Efe,Kl1} one obtains 
similar to \cite[eq. (3.7), (3.15), (4.3)]{KS} that
\begin{align} \notag
\tilde G^{[x]}_{\lambda}(z)
&=\left[ i\int  \psi_{j,\ell} \,\bar\psi_{k,\ell}\; e^{i\BPhi \cdot (z-\frac12(\lb V(x)-A))\BPhi}\:
e^{\left(\frac{i}4\! \sum\limits_{y:d(x,y)=1} \BPhi \cdot \tilde G^{(y|x)}_{\lambda} \BPhi \right)}\;D\BPhi\right]_{j,k} \\
&=-\frac{i}{n}\int  \BPsi^\ot\, e^{i\BPhi \cdot (z-\frac12(\lb V(x)-A))\BPhi}\:
e^{\left(\frac{i}4\! \sum\limits_{y:d(x,y)=1} \BPhi \cdot \tilde G^{(y|x)}_{\lambda} \BPhi \right)}\;D\BPhi
 \label{eq-G}
\end{align}
and
\begin{align}\label{eq-GG}
\left| \tilde G^{[x]}_{\lb}\, (z)\right|^2 & = 
\tfrac{1}{n^2}\int \, \BPsi_-^\ot\;\BPsi_+^\ot\;
e^{i \BPhi_+\cdot(z-\frac12(\lb V(x)+A))\BPhi_+ - i \BPhi_-\cdot(\bar z - \frac12(\lb V(x)+A))\BPhi_-}\; \\
&  \quad   \times e^{\left(\frac{i}{4}\,\sum\limits_{y:d(x,y)=1} (\BPhi_+\cdot \tilde G_\lb^{(y|x)}(z) \BPhi_+\,-\,
\BPhi_- \cdot \overline{\tilde G_\lb^{(y|x)}(z)} \BPhi_-)\right)}\, D\BPhi_+\,D\BPhi_-   \notag\,.
\end{align}

Next, let us take expectation values and use the fact that the matrix potential $V(x)$ is independent identically distributed.
Then one obtains
\begin{align}
\E(G^{[x]}_{\lb}\, (z)) &=
 -\frac{i}{2n} \int   \BPsi^\ot \,e^{i\BPhi\cdot\frac12 (z -A)\BPhi}\, h(\tfrac12 \lb  \BPhi^{\od 2}) 
\prod_{y:d(x,y)=1} \ze^{(y|x)}_{\lb,z} (\BPhi^{\od 2})\,  D\BPhi  \;,
\label{eq-EG-x} \\
\label{eq-EGG-x}
\E\left(\left| G^{[x]}_{\lb}\, (z)\right|^2\right) &= \frac{1}{4n^2}\int\;
 \BPsi_-^\ot\, \BPsi_+^\ot\;
e^{i [\BPhi_+ \cdot\frac12(z-A)\BPhi_+\,-\,\BPhi_-\cdot\frac12(\bar z -A)\BPhi_-]}\;
\\ 
&  \qquad \quad \times\; h(\tfrac12\lambda(\BPhi_+^{\od 2}-\BPhi_-^{\od 2}))\;
\!\!\!\prod_{y:d(x,y)=1}\!\!\! \xi^{(y|x)}_{\lb,z}(\BPhi_+^{\od 2} , \BPhi_-^{\od 2})\;D\BPhi_+\,D\BPhi_- \,. \notag
\end{align}

Moreover, like in \cite[eq. (3.13)]{KS} we get the following relations.
Let $d(x,y)=1$, then one obtains the following.
\begin{equation}
 e^{\frac{i}{4}\,\BPhi'\cdot \tilde G^{(x|y)}_{\lambda}(z)\,\BPhi'} = 
 \int e^{i\BPhi'\cdot\BPhi} e^{i \BPhi \cdot (z-\lambda \frac12 V(x)-\frac 12 A)\BPhi}\Big(
\prod\limits_{\substack{x': d(x,x')=1\\ x'\neq y}} e^{\frac{i}4 \BPhi\cdot \tilde G^{(y|x)}_{\lambda}(z) \BPhi} \Big)\,D\BPhi\;.
\end{equation}
Taking expectations leads to
\begin{align}\label{eq-ze-rec-gen}
 \ze^{(x|y)}_{\lb,z} &=
TB_{\lb,z} \bigg(\prod\limits_{\substack{x': d(x,x')=1\\ x'\neq y}} 
\ze^{(x'|x)}_{\lb,z}
\bigg)\;,\\ \label{eq-xi-rec-gen}
 \xi^{(x|y)}_{\lb,z} &=
\Tt \Bb_{\lb,z} \bigg(\prod\limits_{\substack{x': d(x,x')=1\\ x'\neq y}} 
\xi^{(x'|x)}_{\lb,z}
\bigg)\;.
\end{align}
Similar calculations give \eqref{eq-zeta-recursion} and \eqref{eq-xi-recursion}.

%%%%%%%%%%%%%%%%%%%%%%%%%%%%%%%%%%%%%%%%%%%%%%%%%%%%%%%%%%%%%%%%%%%%%%%%%%%%%%%%%%%%%%%%%%%%%%%%%%%%%%
%%%%%%%%%%%%%%%%%%%%%%%%%%%%%%%%%%%%%%%%%%%%%%%%%%%%%%%%%%%%%%%%%%%%%%%%%%%%%%%%%%%%%%%%%%%%%%%%%%%%%%
%%%%%%%%%%%%%%%%%%%%%%%%%%%%%%%%%%%%%%%%%%%%%%%%%%%%%%%%%%%%%%%%%%%%%%%%%%%%%%%%%%%%%%%%%%%%%%%%%%%%%%
%%%%%%%%%%%%%%%%%%%%%%%%%%%%%%%%%%%%%%%%%%%%%%%%%%%%%%%%%%%%%%%%%%%%%%%%%%%%%%%%%%%%%%%%%%%%%%%%%%%%%%

\end{document}